\documentclass[a4paper,twocolumn,11pt,accepted=2023-07-25]{quantumarticle}
\pdfoutput=1
\usepackage[utf8]{inputenc}
\usepackage[english]{babel}
\usepackage[T1]{fontenc}
\usepackage{amsmath}
\usepackage[numbers]{natbib}

\usepackage{tikz}
\usepackage{lipsum}

\usepackage{amsmath}
\usepackage{amssymb,bm}
\usepackage{amsthm}

\newtheorem{innercustomgeneric}{\customgenericname}
\providecommand{\customgenericname}{}
\newcommand{\newcustomtheorem}[2]{%
	\newenvironment{#1}[1]
	{%
		\renewcommand\customgenericname{#2}%
		\renewcommand\theinnercustomgeneric{##1}%
		\innercustomgeneric
	}
	{\endinnercustomgeneric}
}

\newcustomtheorem{customthm}{Theorem}
\newcustomtheorem{definitionBob}{Definition}

\usepackage{color,dsfont} 
\usepackage{graphicx}
\usepackage{subcaption}
\usepackage{ragged2e}
\DeclareCaptionJustification{justified}{\justifying}
\usepackage[colorlinks=true, hyperindex, breaklinks, linkcolor=blue, urlcolor=blue, citecolor=blue]{hyperref} 
\usepackage[capitalise]{cleveref}
\usepackage{mathrsfs}

\usepackage{mathtools}
\usepackage{float}
\usepackage{verbatim}
\usepackage{latexsym}
\usepackage{amsmath}
\usepackage{amssymb}
\usepackage{setspace}
\usepackage{amsfonts}
\usepackage{stmaryrd}
\usepackage{xcolor}
\usepackage{enumitem}
\usepackage{hhline}
\usepackage[normalem]{ulem}

\newtheorem{theorem}{Theorem} 
 
 \newtheorem{definition}{Definition}

\newtheorem{algorithm}{Algorithm} \newtheorem{protocol}{Protocol}

\newcommand{\codepar}[1]{\ensuremath{[\![#1]\!]}}

\renewcommand\figurename{FIGURE}
\crefname{equation}{Eq.\!}{Eqs.\!}
\crefname{figure}{Fig.\!}{Figs.\!}
\usepackage{multirow}
\mathchardef\mhyphen="2D

\begin{document}
	
\title{Adaptive syndrome measurements for Shor-style error correction}
	
\author{Theerapat Tansuwannont}
\email{t.tansuwannont@duke.edu}
\affiliation{
	Duke Quantum Center, Duke University, Durham, NC 27701, USA
}
\affiliation{
	Department of Electrical and Computer Engineering, Duke University, Durham, NC 27708, USA
}
\orcid{0000-0002-2865-0705}
	
\author{Balint Pato}
\email{balint.pato@duke.edu}
\affiliation{
	Duke Quantum Center, Duke University, Durham, NC 27701, USA
}
\affiliation{
	Department of Electrical and Computer Engineering, Duke University, Durham, NC 27708, USA
}
\orcid{0000-0001-9502-3368}
	
\author{Kenneth R. Brown}
\email{ken.brown@duke.edu}
\affiliation{
	Duke Quantum Center, Duke University, Durham, NC 27701, USA
}
\affiliation{
	Department of Electrical and Computer Engineering, Duke University, Durham, NC 27708, USA
}
\affiliation{
	Department of Physics, Duke University, Durham, NC 27708, USA
}
\affiliation{
	Department of Chemistry, Duke University, Durham, NC 27708, USA
}
\orcid{0000-0001-7716-1425}
	
\maketitle
	
\begin{abstract}
The Shor fault-tolerant error correction (FTEC) scheme uses transversal gates and ancilla qubits prepared in the cat state in syndrome extraction circuits to prevent propagation of errors caused by gate faults. For a stabilizer code of distance $d$ that can correct up to $t=\lfloor(d-1)/2\rfloor$ errors, the traditional Shor scheme handles ancilla preparation and measurement faults by performing syndrome measurements until the syndromes are repeated $t+1$ times in a row; in the worst-case scenario, $(t+1)^2$ rounds of measurements are required.
In this work, we improve the Shor FTEC scheme using an adaptive syndrome measurement technique. The syndrome for error correction is determined based on information from the differences of syndromes obtained from consecutive rounds.
Our protocols that satisfy the strong and the weak FTEC conditions require no more than $(t+3)^2/4-1$ rounds and $(t+3)^2/4-2$ rounds, respectively, and are applicable to any stabilizer code. 
Our simulations of FTEC protocols with the adaptive schemes on hexagonal color codes of small distances verify that our protocols preserve the code distance, can increase the pseudothreshold, and can decrease the average number of rounds compared to the traditional Shor scheme. We also find that for the code of distance $d$, our FTEC protocols with the adaptive schemes require no more than $d$ rounds on average.
\end{abstract}



\section{Introduction}
\label{sec:intro}

One essential component for constructing a large-scale quantum computer is quantum error correction (QEC). One has to make sure that the QEC process can be implemented fault-tolerantly so that a small number of faults in the process will not cause uncorrectable errors. It has been proved that a fault-tolerant error correction (FTEC) scheme and other schemes for fault-tolerant quantum computation (FTQC) can be used to simulate any quantum circuit with an arbitrarily low logical error rate if the physical error rate is below some scheme-dependent threshold value \cite{Shor96,AB08,Kitaev97,KLZ96,Preskill98,TB05,ND05,AL06,AGP06}. However, implementing an FTEC scheme is physically challenging because larger space and time overhead (ancilla qubits and quantum gates) are required for a lower logical error rate \cite{Steane03,PR12,CJL16b,TYC17}. Another reason is that for the same family of quantum error correcting codes (QECC), an FTEC protocol that requires more space and time overhead tends to have a lower fault-tolerant threshold since there are more possible fault combinations that can cause the protocol to fail \cite{AGP06}. Therefore, an FTEC scheme that requires only a small amount of overhead is desirable. 

The Shor FTEC scheme \cite{Shor96} is one of the very first FTEC schemes. It handles gate faults by measuring the stabilizer generators of a stabilizer code using the cat states and transversal gates. In addition, the Shor FTEC scheme handles ancilla preparation and measurement faults by repeated syndrome measurements; traditionally, full syndrome measurements are performed until the outcomes are repeated $t+1$ times in a row, where $t=\lfloor (d-1)/2 \rfloor$ is the number of errors that a stabilizer code of distance $d$ can correct. The Shor FTEC scheme satisfies the strong FTEC conditions \cite{AGP06} (to be defined in \cref{def:strong_FT}), so it is compatible with code concatenation (an FTEC scheme that only satisfies the weak FTEC conditions \cite{DR20}, defined in \cref{def:weak_FT}, is not compatible with code concatenation and only works at the highest level of concatenation). The Shor scheme is applicable to any stabilizer code and requires the number of ancilla qubits to be equal to the maximum weight of the stabilizer generators. In the worst-case scenario (when there are no more than $t$ faults in the scheme), the Shor scheme requires $(t+1)^2$ rounds of full syndrome measurements.

One way to reduce the time overhead required for an FTEC scheme similar to the Shor scheme is by using adaptive syndrome measurements, in which the measurement sequences depend on the prior measurement outcomes. Zalka \cite{Zalka96} constructed an adaptive Shor-style FTEC protocol for the \codepar{7,1,3} Steane code that uses Shor circuit for syndrome extraction. Delfosse and Reichardt \cite{DR20} further developed the adaptive measurement idea and constructed adaptive Shor-style FTEC protocols for any stabilizer code of distance 3, any Calderbank-Shor-Steane (CSS) codes of distance 3, and some stabilizer codes of distance $\leq 7$. Their protocols for different code families require different maximum numbers of syndrome bit measurements. One drawback of the FTEC protocols in \cite{DR20} is that they only satisfy the weak FTEC conditions and are not generally compatible with code concatenation (unless the code being used is a perfect code, a perfect CSS code, or the \codepar{16,4,3} color code invented in \cite{Reichardt18}). 


In this work, we develop Shor-style FTEC schemes with adaptive syndrome measurements that use information from the differences of syndromes obtained from any two consecutive rounds. The main results of this work are the following: (1) we construct an adaptive Shor-style FTEC scheme that satisfies the strong FTEC conditions (\cref{def:strong_FT}) and is applicable to \emph{any} stabilizer code. In our protocol, stabilizer generators are measured using Shor syndrome extraction circuits. For a protocol that can tolerate up to $t=\lfloor (d-1)/2 \rfloor$ faults (which is applicable to a stabilizer code of distance $d$), we prove that the number of required rounds of full syndrome measurements is no more than $(t+3)^2/4-1$. We also discuss some minor improvements that can further reduce the total number of syndrome bit measurements. (2) We construct an adaptive Shor-style FTEC scheme that satisfies the weak FTEC conditions (\cref{def:weak_FT}) and is applicable to any stabilizer code. The protocol requires at most $(t+2)^2/4$ rounds if the syndrome obtained from the first round is $\vec{s}_1\neq 0$, and requires at most $(t+3)^2/4-2$ rounds if $\vec{s}_1=0$. With some minor improvements, our protocol that satisfies the weak conditions and can correct up to 1 fault is similar to the adaptive FTEC protocol for a distance-3 stabilizer code proposed in \cite{DR20}. Our main results on the maximum number of rounds are summarized in \cref{tab:main_results}.
(3) We show that the lower bound of the fault-tolerant threshold for a concatenated code can be improved when the adaptive scheme satisfying the strong FTEC conditions is used instead of the traditional Shor scheme. We also compare both of our adaptive schemes with the traditional Shor scheme by simulating FTEC protocols on the hexagonal color codes of distance 3, 5, 7, and 9. Our numerical results verify that the adaptive schemes are fault tolerant, preserve the code distance, can increase the pseudothreshold, and can decrease the average number of rounds of syndrome measurements. We also observe that the average number of rounds for the strong and the weak schemes approach $d$ and $d-1$, respectively, as the physical error rate approaches $1$. 


\begin{table}[tbp]
	\begin{center}
		\begin{tabular}{| c | c |}
			\hline
			\multirow{2}{*}{Protocol} & The maximum number \\
			& of required rounds\\ 
			\hline
			Strong FTEC & $(t+3)^2/4-1$  \\
			\hline
			\multirow{2}{*}{Weak FTEC} & $(t+2)^2/4$ if $\vec{s}_1\neq 0$\\ \cline{2-2}
			& $(t+3)^2/4-2$ if $\vec{s}_1=0$\\
			\hline
			Traditional Shor \cite{Shor96} & $(t+1)^2$ \\
			\hline
		\end{tabular}
	\end{center}
	\caption{The maximum number of rounds required for the protocols in this work which satisfy the strong FTEC conditions (\cref{def:strong_FT}) and the weak FTEC conditions (\cref{def:weak_FT}) compared to the traditional Shor FTEC protocol when applying to a stabilizer code of distance $d$ that can correct up to $t=\lfloor (d-1)/2 \rfloor$ errors. Our protocols are applicable to any stabilizer code.}
	\label{tab:main_results}
\end{table}

This paper is organized as follows: In \cref{sec:Shor}, we formally define the strong and weak conditions for FTEC, and briefly review the traditional Shor FTEC scheme. In \cref{sec:strong_EC}, we introduce the notion of difference vector, construct an algorithm to find an error syndrome suitable for FTEC, and construct an FTEC protocol that satisfies the strong conditions. In \cref{sec:weak_EC}, we apply the algorithm from the previous section and construct an FTEC protocol that satisfies the weak conditions. 
In \cref{sec:comparison}, we compare our adaptive schemes developed in the previous sections with the traditional Shor scheme both analytically and numerically.
Our results and possible future directions are discussed in \cref{sec:discussion}.

\section{FTEC conditions and the traditional Shor FTEC scheme}
\label{sec:Shor}






A quantum \codepar{n,k,d} stabilizer code \cite{Gottesman96,Gottesman97} uses $n$ physical qubits to encode $k$ logical qubits and can correct an error of weight up to $\tau = \lfloor(d-1)/2\rfloor$, where $d$ is the code distance. A stabilizer code can be described by its corresponding stabilizer group, the Abelian group generated by $r=n-k$ commuting independent Pauli operators called stabilizer generators. The coding subspace is a simultaneous $+1$ eigenspace of all elements in the stabilizer group. In an ideal situation, if the weight of the error on a codeword is no more than $\tau$, a process of quantum error correction (QEC) can remove such an error. For a stabilizer code, the QEC process involves measuring the eigenvalues of all stabilizer generators on the corrupted codeword. The combined measurement results, called error syndrome, will be used to find an EC operator for undoing the corruption. 

In practice, however, any quantum gate involved in the syndrome measurements can be faulty. In this work, we will assume the standard \emph{depolarizing noise model} in which each one-qubit gate is followed by a single-qubit Pauli error ($I, X, Y,$ or $Z$), each two-qubit gate is followed by a two-qubit Pauli error of the form $P_1\otimes P_2$ (where $P_1,P_2 \in \{I,X,Y,Z\}$), and each single qubit measurement (which outputs a classical bit of information) is followed by either no error or a bit-flip error. Note that an error from each fault in the QEC process may propagate to other qubits (depending on the circuit being used in the syndrome measurement) and become an error of higher weight on the data block; i.e., a few faults may lead to the total error of weight more than $\tau$, causing the QEC process to fail. To prevent such cases from happening, we want to make sure that the QEC protocol being used is \emph{fault tolerant}; vaguely speaking, if the weight of an input error plus the number of faults in the FTEC protocol is small enough, we want to make sure that the output state is logically correct and has an error of weight no more than the total number of faults in the FTEC protocol. For a stabilizer code that can correct errors up to weight $\tau$, one might want to construct an FTEC protocol that can tolerate up to $t$ faults where $t$ is as close to $\tau$ as possible.


We can define the strong conditions for FTEC as follows:

\begin{definition}{Strong conditions for fault-tolerant error correction \cite{AGP06}}
	
	Let $t \leq \lfloor (d-1)/2\rfloor$ where $d$ is the distance of a stabilizer code. An error correction protocol is \emph{strongly $t$-fault tolerant} if the following two conditions are satisfied:
	\begin{enumerate}
		\item Error correction correctness property (ECCP): For any input codeword with error of weight $r$, if $s$ faults occur during the protocol with $r+s \leq t$, ideally decoding the output state gives the same codeword as ideally decoding the input state.
		\item Error correction recovery property (ECRP): If $s$ faults occur during the protocol with $s \leq t$, regardless of the weight of the error on the input state, the output state differs from any valid codeword by an error of weight at most $s$.
	\end{enumerate}
	\label{def:strong_FT}
\end{definition}

(Note that \cref{def:strong_FT} can be further generalized by defining $r$ as the number of faults that causes the input error, as proposed in \cite{TL22}. For an FTEC scheme in which stabilizer generators are measured using cat states and transversal gates similar to the Shor FTEC scheme, however, there is no difference in the uses of \cref{def:strong_FT} and the generalized definition in \cite{TL22} since a single gate fault in each generator measurement will lead no more than weight 1 error on the data qubits.)

\cref{def:strong_FT} is one of the main ingredients to prove the \emph{threshold theorem} in \cite{AGP06}, a theorem which shows that an arbitrarily low logical error rate can be attained through code concatenation if the physical error rate is below some threshold value.
It should be noted that for an FTEC protocol satisfying the strong FTEC conditions when the weight of the input error is large ($r+s >t$) but the number of faults is small ($s \leq t$), ECRP guarantees that the output state will be in the `correctable' subspace, but the input and the output states might not be logically the same. This property is necessary for constructing a conventional FTEC protocol for a concatenated code; for a code with two levels of concatenation, an FTEC protocol for the 2nd-level code is constructed from an FTEC protocol for the 1st-level code by replacing every physical qubit with a code block and replacing every physical gate with its corresponding logical gate. To do error correction, the 1st-level FTEC protocol is applied on each code block, and the 2nd-level FTEC protocol is applied afterwards. ECRP guarantees that an error on each code block after applying the 1st-level protocol can be corrected by the 2nd-level FTEC protocol. The idea can also be extended to a code with more levels of concatenation; see \cite{AGP06}.


For some families of codes, a code of high distance can be obtained without code concatenation. Surface codes \cite{Kitaev97,BK98} and color codes \cite{BM06} are examples of topological codes in which the code distance can be made arbitrarily large by increasing the lattice size. For such code families, an arbitrarily low logical error rate can be attained without code concatenation if the physical error rate is below some threshold value. In that case, there is no need to guarantee the weight of the output error for high-weight input errors with $r+s > t$. To achieve fault tolerance, it is sufficient to show that an FTEC protocol for such code families satisfies the following weak conditions for FTEC:

\begin{definition}{Weak conditions for fault-tolerant error correction \cite{DR20}}
	
	Let $t \leq \lfloor (d-1)/2\rfloor$ where $d$ is the distance of a stabilizer code. An error correction protocol is \emph{weakly $t$-fault tolerant} if the following two conditions are satisfied:
	\begin{enumerate}
		\item ECCP: For any input codeword with error of weight $r$, if $s$ faults occur during the protocol with $r+s \leq t$, ideally decoding the output state gives the same codeword as ideally decoding the input state.
		\item ECRP: For any input codeword with error of weight $r$, if $s$ faults occur during the protocol with $r+s \leq t$, the output state differs from any valid codeword by an error of weight at most $s$.
	\end{enumerate}
	
	\label{def:weak_FT}
\end{definition}



The Shor FTEC scheme \cite{Shor96} is an example of an FTEC scheme that satisfies the strong FTEC conditions. The details of the traditional Shor scheme are as follows: Suppose that a stabilizer generator being measured is a Pauli operator of weight $w$ of the form $M = P_1\otimes P_2 \otimes \cdots \otimes P_w$. An eigenvalue of the stabilizer generator is measured by first preparing ancilla qubits in a \emph{cat state} of the form $\frac{1}{\sqrt{2}}(|0\rangle^{\otimes w}+|1\rangle^{\otimes w})$, then applying controlled-$P_1$, controlled-$P_2$, ..., controlled-$P_w$ gates; see \cref{fig:Shor} for an example. Afterward, Hadamard gates are applied transversally to the ancilla qubits, which are measured in the $Z$ basis. The even and odd parities of the measurement results of ancilla qubits correspond to $+1$ and $-1$ eigenvalues of $M$, respectively. For convenience, we will call a circuit for measuring an eigenvalue of a stabilizer generator in this form the \emph{Shor syndrome extraction circuit}. (Note that the cat state used in the Shor syndrome extraction circuit must be prepared fault-tolerantly; i.e., if there are $s\leq t$ faults during the cat state preparation, the resulting cat state must differ from an ideal cat state by an error of weight no more than $s$. This can be done by using the ancilla verification method in \cite{Shor96} or the ancilla decoding and measurement method in \cite{DA07}.)


Since some faults may lead to an incorrect measurement outcome, the full syndrome measurement will be performed repeatedly. In the traditional Shor FTEC scheme, the syndromes will be measured until they are repeated $t+1$ times in a row. By doing so, we can make sure that if there are no more than $t$ faults in the whole protocol, there is at least one correct round in the last $t+1$ rounds with the same syndrome, and the repeated syndrome corresponds to the data error at the end of the correct round. An EC operator to be applied is a Pauli operator of the minimum weight whose syndrome is the repeated syndrome. Here we will call the process of selecting the syndrome for error correction using the aforementioned criteria \emph{Shor decoder}.


\begin{figure}[tbp]
	\centering
	\includegraphics[width=0.45\textwidth]{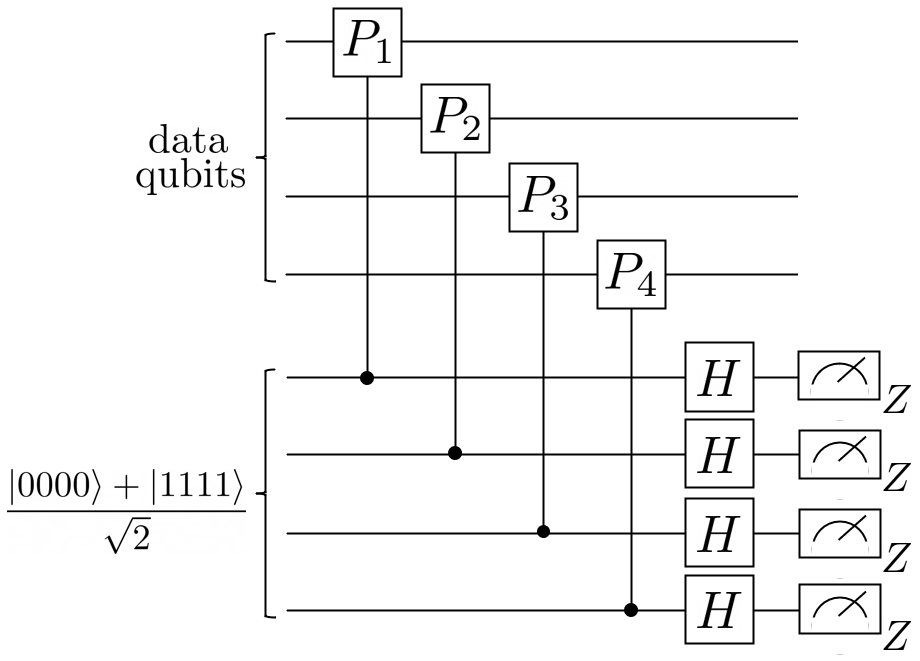}
	\caption{Shor syndrome extraction circuit for measuring a stabilizer generator of the form $M=P_1\otimes P_2 \otimes P_3 \otimes P_4$. The ancilla qubits are initially prepared in the cat state $\frac{1}{\sqrt{2}}(|0000\rangle+|1111\rangle)$ and measured in the $Z$ basis at the end. Even and odd parities of the measurement results correspond to $+1$ and $-1$ eigenvalues of $M$.}
	\label{fig:Shor}
\end{figure}


It is not hard to verify that both conditions in \cref{def:strong_FT} are satisfied with $t=\tau=\lfloor (d-1)/2\rfloor$; because the controlled Pauli gates are applied transversally between the data block and the ancilla qubits, each single gate fault will lead to an error of weight $\leq 1$ on the data block or the ancilla qubits (or both). Also, any errors on the ancilla qubits and measurement faults can be handled by repeatedly performing syndrome measurements. When $s \leq t$, the repeated syndrome is the syndrome of the input error plus any error that occurs before the last correct round in the last $t+1$ rounds with the same syndrome. Therefore, after applying the EC operator corresponding to the repeated syndrome, the output state differs from an uncorrupted logical state by an error of weight $\leq s$ (an error that may arise from some faults after the last correct round). The input and the output states are always logically the same when $r+s \leq t$, but the states may be logically different when $r+s > t$. For these reasons, both ECCP and ECRP in \cref{def:strong_FT} are satisfied.

Any quantum circuit can be fault-tolerantly simulated with arbitrarily low logical error rates using an FTEC scheme together with fault-tolerant gadgets for quantum gates, state preparation, and state measurement. However, a lower logical error rate requires more overhead (ancilla qubits and quantum gates). One drawback of the traditional Shor scheme is that the number of required ancilla qubits is equal to the maximum weight of stabilizer generators; this is because of the Shor syndrome extraction circuit. Another drawback is that the traditional Shor scheme requires repeated syndrome measurements. Suppose that there are $t$ faults in the protocol, in the worst-case scenario, $(t+1)^2$ rounds of the syndrome measurements must be performed before the syndromes are repeated $t+1$ times in a row; this is because of the Shor decoder.

There are several FTEC schemes whose syndrome extraction circuits require fewer ancilla qubits compared to the Shor syndrome extraction circuit. Examples of such FTEC schemes are the flag FTEC schemes, in which a few `flag' ancilla qubits are used to detect faults that can lead to a high-weight error on the data block \cite{CR17a}. The flag FTEC scheme for a general stabilizer code of distance $d$ requires $d+1$ ancillas \cite{CR20}, while the schemes for certain families of codes may require fewer \cite{CR17a,CB18,TCL20,CKYZ20,CZYHC20,TL21,TL22}. However, to handle syndrome measurement faults, the flag scheme still requires repeated syndrome measurements which use ideas similar to the Shor decoder.



The main goal of this work is to construct a better algorithm for finding a syndrome suitable for FTEC which requires fewer rounds of syndrome measurements compared to the Shor decoder in the traditional Shor scheme (where the syndrome measurements are performed until the syndromes are repeated $t+1$ times in a row). In \cref{sec:strong_EC,sec:weak_EC}, we will develop FTEC protocols satisfying the strong FTEC conditions (\cref{def:strong_FT}) and the weak FTEC conditions (\cref{def:weak_FT}), respectively. 
In our protocols, we will assume that stabilizer generators are measured using the Shor syndrome extraction circuits and focus on reducing the number of required rounds. 
The processes of selecting the syndrome for error correction in these protocols will be referred to as \emph{adaptive strong and adaptive weak decoders}.






\section{Adaptive measurements for Shor error correction satisfying the strong FTEC conditions}
\label{sec:strong_EC}

In this section, we will construct an FTEC protocol that satisfies the strong FTEC conditions (\cref{def:strong_FT}). Stabilizer generators will be measured using the Shor syndrome extraction circuits, so any single gate fault will cause an error of weight $\leq 1$ on the data block or one of the ancilla qubits (or both). The main difference between our protocol and the traditional Shor FTEC scheme is that we will perform the syndrome measurements in an `adaptive' way; instead of measuring until the syndromes are repeated $t+1$ times in a row (where $t$ is the number of faults that the protocol can correct), the condition to stop the measurement sequences will change dynamically depending on syndromes collected from all rounds. We call this kind of procedure \emph{adaptive measurements} because of its similarity to the measurement procedure proposed by Delfosse and Reichardt in \cite{DR20} (later in \cref{sec:weak_EC}, readers will find that our FTEC protocol satisfying the weak FTEC conditions is similar to the FTEC protocol in \cite{DR20} when applying to a stabilizer code of distance 3).



To make sure that both ECCP and ECRP in \cref{def:strong_FT} are satisfied, we will use the following ideas: given the whole syndrome history, we will try to find a syndrome $\vec{s}_i$ obtained from round $i$ that is \emph{correct}, i.e., it corresponds to the data error at the end of round $i$. Finding such a syndrome should be possible regardless of the weight of the input error. Let $r$ be the weight of the input error, $s$ be the number of faults in the protocol, and $t=\lfloor (d-1)/2\rfloor$ be the weight of error that a stabilizer code of distance $d$ can correct. Suppose that $s \leq t$. If such a syndrome $\vec{s}_i$ can be found, combining the input error, the error from faults occurred up to round $i$, and the EC operator corresponding to $\vec{s}_i$ will result in a logical operator; it is always a trivial logical operator (i.e., a stabilizer) when $r+s \leq t$ and it can be a nontrivial logical operator when $r+s > t$. After applying the EC operator, the output state will differ from an uncorrupted logical state by an error from faults that occur after round $i$ (the error weight is always $\leq s$). If the procedure explained above can be done, both ECCP and ECRP in \cref{def:strong_FT} are satisfied.


\subsection{Difference vectors for single-fault cases and an FTEC protocol satisfying the strong FTEC conditions for a stabilizer code of distance 3}
\label{subsec:diff_single}



Our algorithm for finding a syndrome suitable for error correction will use the information from the \emph{differences of syndromes between any two consecutive rounds}. First, let us consider how a single fault can affect the differences between the syndrome from the round that the fault occurs, and the syndromes from the rounds before and after.



Suppose that in each round of full syndrome measurements, stabilizer generators are measured sequentially. Let $\vec{s}_j$ denote the syndrome obtained from the $j$-th round of full syndrome measurements, and assume that a single fault occurs on the $i$-th round (an input error of weight 1 can be considered as a single data-qubit fault on the $0$-th round).
\begin{enumerate}
	\item Let $E$ denote the data error at the end of the $(i-1)$-th round. If a single fault during a generator measurement on the $i$-th round causes a data error $F$, subsequent generator measurements in the same round may or may not be able to detect the newly occurred error. $F$ may be fully detectable ($\vec{s}_{i}$ is exactly the syndrome of $E\cdot F$), partially detectable (some part of $\vec{s}_{i}$ represents the syndrome of $E\cdot F$, and the other part represents the syndrome of $E$), or undetectable ($\vec{s}_{i}$ represents the syndrome of $E$ only). Thus, $\vec{s}_{i}$ may or may not be the syndrome of the data error at the end of the $i$-th round. Nevertheless, $F$ will be fully detectable by the syndrome measurements at the $(i+1)$-th round, so $\vec{s}_{i+1}$ is the syndrome of the data error at the end of the $i$-th round (which is the syndrome of $E\cdot F$ in this case).
	\item If a single fault during a generator measurement causes an error on the ancilla qubits or is an ancilla measurement fault, it may cause a single bit-flip on $\vec{s}_{i}$, so $\vec{s}_{i}$ may or may not be the syndrome of the data error at the end of the $i$-th round. Nevertheless, $\vec{s}_{i+1}$ is the syndrome of the data error at the end of the $i$-th round.
\end{enumerate}
Possible single faults can be categorized into three types depending on their effects on the syndromes:
\begin{enumerate}
	\item Type I: a single fault on the $i$-th round that causes $\vec{s}_{i-1}\neq \vec{s}_i \neq \vec{s}_{i+1}$. Examples of Type I faults are an ancilla measurement fault and a fault leading to a data error partially detectable by generator measurements in the $i$-th round.
	\item Type II: a single fault on the $i$-th round that causes $\vec{s}_{i-1}= \vec{s}_i \neq \vec{s}_{i+1}$. An example of a Type II fault is a fault leading to a data error undetectable by generator measurements in the $i$-th round.
	\item Type III: a single fault on the $i$-th round that causes $\vec{s}_{i-1}\neq \vec{s}_i = \vec{s}_{i+1}$. An example of a Type III fault is a fault leading to a data error fully detectable by generator measurements in the $i$-th round.
\end{enumerate}
Since for any Type III fault on the $i$-th round, there is a Type II fault on the $(i-1)$-th round that causes the same data error, it is safe to consider only faults of Types I and II (a fully detectable data error from a fault on the 1st round is equivalent to an input error). Note that any single fault on the $i$-th round cannot cause $\vec{s}_{i-1}=\vec{s}_i=\vec{s}_{i+1}$ unless the data error is trivial. This is because the $(i+1)$-th round of syndrome measurements can always detect a data error of weight 1 from the $i$-th round when the code distance is $d \geq 3$. This is also true in the case of multiple faults because a data error of weight $\leq t$ from the the $i$-th round are always detectable by the $(i+1)$-th round of syndrome measurements when the code distance is $d \geq 2t+1$, useless the data error is trivial.

For convenience, we will define a difference vector from a sequence of syndrome measurement results as follows:
\begin{definition}{Difference vector}
	
	Let $m$ be the total number of rounds of full syndrome measurements, and let $\vec{s}_i$ denote the error syndrome obtained from the $i$-th round. The difference vector $\vec{\delta}$ is an $(m-1)$-bit string in which the $i$-th bit $\delta_i$ is 0 if $\vec{s}_{i+1}=\vec{s}_i$, or $\delta_i$ is 1 if $\vec{s}_{i+1}\neq\vec{s}_i$.
	\label{def:dif_vec}
\end{definition}
Let $\mathrm{I}(i)$ and $\mathrm{II}(i)$ denote single faults of Types I and II on the $i$-th round, where $i=1,\dots,m$ and $m$ is the total number of rounds. By \cref{def:dif_vec}, the difference vector of length $m-1$ corresponding to each fault type is the following:
\begin{enumerate}
	\item For $\mathrm{I}(1)$, $\vec{\delta}=1\;0\dots 0\;0$.
	\item For $\mathrm{I}(i)$ ($i\neq 0$ or $m$), $\vec{\delta}=0\dots 0\underset{i-1}{1}\underset{i}{1}\;0\dots 0$.
	\item For $\mathrm{I}(m)$, $\vec{\delta}=0\;0\dots 0\;1$.
	\item For $\mathrm{II}(i)$ ($i=1,\dots,m-1$), $\vec{\delta}=0\dots 0\;\underset{i}{1}\;0\dots 0$.
	\item For $\mathrm{II}(m)$, $\vec{\delta}=0\;0\dots 0\;0$.
\end{enumerate}
(For an input error which may be denoted by $\mathrm{II}(0)$, the difference vector is the zero vector.)

To see how a difference vector can be used to determine a syndrome suitable for error correction, let us consider an FTEC protocol correcting up to $t=1$ fault as an example. When only $\mathrm{I}(i)$ occurs, the syndrome $\vec{s}_i$ is the only syndrome that cannot be used for error correction since it might not correspond to the data error at the end of any round. On the other hand, when only $\mathrm{II}(i)$ occurs, the syndrome from any round can be used for error correction; $\vec{s}_{i}=\vec{s}_{i-1}=\dots$ corresponds to the data error at the end of the $(i-1)$-th round (which is trivial), while $\vec{s}_{i+1}=\vec{s}_{i+2}=\dots$ corresponds to the data error at the end of the $i$-th round (which is the data error caused by $\mathrm{II}(i)$). 

In actual syndrome measurements, we cannot perfectly distinguish between Type I and Type II faults using the difference vector as some faults of different types can give the same difference vector (for example, $\mathrm{I}(1)$ and $\mathrm{II}(1)$, and $\mathrm{I}(m)$ and $\mathrm{II}(m-1)$). Nevertheless, whenever we find $\delta_i=0$, we are certain that neither Type I nor Type II fault occurs on the $i$-th round. That is, $\vec{s}_i$ is \emph{usable} for error correction if $\delta_i=0$ ($\vec{s}_{i+1}=\vec{s}_i$ is also usable). Another case that a usable syndrome can be found is whenever $\vec{\delta}$ has a substring 11, which implies that one fault already occurred on some round before the $m$-th round (the latest round). In that case, we can do error correction using the syndrome obtained from the $m$-th round. Using these facts, an FTEC protocol satisfying the strong FTEC conditions with $t=1$ can be constructed as follows:

\begin{protocol}{FTEC protocol satisfying the strong FTEC conditions for a stabilizer code of distance 3}
	
	In each round of full syndrome measurements, measure stabilizer generators using the Shor syndrome extraction circuits. After the $j$-th round ($j\geq 2$), calculate the $(j-1)$-th bit of the difference vector. Repeat syndrome measurements until one of the following conditions is satisfied, then perform error correction using the error syndrome corresponding to each condition:
	\begin{enumerate}
		\item If $\delta_i=0$ is found after the $(i+1)$-th round, stop the syndrome measurements. Perform error correction using the syndrome $\vec{s}_i$.
		\item If $\vec{\delta}$ contains a substring $11$, stop the syndrome measurements. Perform error correction using the syndrome obtained from the latest round.
	\end{enumerate}  
	\label{pro:strong_t_1} 
\end{protocol}

\begin{table*}[tbp]
	\begin{center}
		\begin{tabular}{| c | c | c | c |}
			\hline
			\multirow{2}{*}{Fault type} & difference vector & syndromes suitable & syndrome to be used \\
			& $\vec{\delta}$ &  for error correction & for error correction \\
			\hline
			Input error &  $0\;0$ & $\vec{s}_1,\vec{s}_2,\vec{s}_3$ & $\vec{s}_1$ \\
			\hline
			$\mathrm{I}(1)$ &  $1\;0$ & $\vec{s}_2,\vec{s}_3$ & $\vec{s}_2$ \\
			\hline
			$\mathrm{I}(2)$ &  $1\;1$ & $\vec{s}_1,\vec{s}_3$ & $\vec{s}_3$ \\
			\hline
			$\mathrm{I}(3)$ &  $0\;1$ & $\vec{s}_1,\vec{s}_2$ & $\vec{s}_1$ \\
			\hline
			$\mathrm{II}(1)$ &  $1\;0$ & $\vec{s}_1,\vec{s}_2,\vec{s}_3$ & $\vec{s}_2$ \\
			\hline
			$\mathrm{II}(2)$ &  $0\;1$ & $\vec{s}_1,\vec{s}_2,\vec{s}_3$ & $\vec{s}_1$ \\
			\hline
			$\mathrm{II}(3)$ &  $0\;0$ & $\vec{s}_1,\vec{s}_2,\vec{s}_3$ & $\vec{s}_1$ \\
			\hline
		\end{tabular}
	\end{center}
	\caption{All possible single faults, their corresponding difference vectors, syndromes suitable for error correction, and the syndromes that will be used for error correction according to \cref{pro:strong_t_1}, assuming that full syndrome measurements are performed 3 rounds in total.}
	\label{table:t_1}
\end{table*}

Suppose that the total number of rounds in the protocol is 3. All possible single faults, their corresponding difference vectors, syndromes suitable for error correction, and the syndromes that will be used for error correction according to our protocol are displayed in \cref{table:t_1}. In fact, 3 is the smallest number of rounds required to make sure that a usable syndrome exists; 2 rounds are not sufficient since $\mathrm{I}(1)$ and $\mathrm{I}(2)$ give the same $\vec{\delta}=1$ but they cannot be distinguished, and neither $\vec{s}_1$ nor $\vec{s}_2$ works for both cases. 3 is also the number of rounds of syndrome measurements in the worst-case scenario of our protocol for $t=1$; i.e., the total number of rounds is at most 3 in any case.





\subsection{Difference vectors for multiple-fault cases and an FTEC protocol satisfying the strong FTEC conditions for a stabilizer code of any distance}
\label{subsec:diff_multiple}



In this section, we will extend our method for finding a syndrome suitable for error correction to the case of multiple faults so that an FTEC protocol satisfying the strong FTEC conditions for a stabilizer code of any distance can be constructed. First, let us consider the case that up to $t$ faults simultaneously occur on the $i$-th round. Unless the total data error is trivial, a combination of such faults will result in $\vec{s}_{i-1}\neq \vec{s}_i \neq \vec{s}_{i+1}$ (equivalent to a single $\mathrm{I}(i)$ fault), $\vec{s}_{i-1}= \vec{s}_i \neq \vec{s}_{i+1}$ (equivalent to a single $\mathrm{II}(i)$ fault), or $\vec{s}_{i-1}\neq \vec{s}_i = \vec{s}_{i+1}$ (equivalent to a single $\mathrm{II}(i-1)$ fault). A syndrome suitable for error correction, in this case, is similar to that of a Type I or a Type II fault. In other words, what matters is \emph{the presence of any faults in each round.} If we can deal with any case of a single fault, we can also deal with any case that multiple faults occur on the same round. As we aim to analyze the worst-case scenario, we can assume that no more than one fault occurs on each round when considering the case of multiple faults.


Next, we will see how difference vectors of two faults that occur on different rounds can be combined. Let us consider syndromes from any two consecutive rounds $j$ and $j+1$ which arise from two faults $\lambda_A$ and $\lambda_B$. $\lambda_A$ can cause either $\vec{s}_{A,j}=\vec{s}_{A,j+1}$ ($\delta_{A,j}=0$) or $\vec{s}_{A,j}\neq \vec{s}_{A,j+1}$ ($\delta_{A,j}=1$), and $\lambda_B$ can cause either $\vec{s}_{B,j}=\vec{s}_{B,j+1}$ ($\delta_{B,j}=0$) or $\vec{s}_{B,j}\neq \vec{s}_{B,j+1}$ ($\delta_{B,j}=0$). Combining $\lambda_A$ and $\lambda_B$ results in one of the following cases:
\begin{enumerate}
	\item If $s_{A,j}=s_{A,j+1}$ and $s_{B,j}=s_{B,j+1}$, then $s_{A,j}+s_{B,j}=s_{A,j+1}+s_{B,j+1}$; that is, $\delta_{A,j}=0$ and $\delta_{B,j}=0$ lead to $\delta_{AB,j}=0$.
	\item If $s_{A,j}=s_{A,j+1}$ and $s_{B,j}\neq s_{B,j+1}$, then $s_{A,j}+s_{B,j}\neq s_{A,j+1}+s_{B,j+1}$; that is, $\delta_{A,j}=0$ and $\delta_{B,j}=1$ lead to $\delta_{AB,j}=1$.
	\item If $s_{A,j}\neq s_{A,j+1}$ and $s_{B,j}=s_{B,j+1}$, then $s_{A,j}+s_{B,j}\neq s_{A,j+1}+s_{B,j+1}$; that is, $\delta_{A,j}=1$ and $\delta_{B,j}=0$ lead to $\delta_{AB,j}=1$.
	\item If $s_{A,j}\neq s_{A,j+1}$ and $s_{B,j}\neq s_{B,j+1}$, then either $s_{A,j}+s_{B,j}\neq s_{A,j+1}+s_{B,j+1}$ or $s_{A,j}+s_{B,j}=s_{A,j+1}+s_{B,j+1}$; that is, $\delta_{A,j}=1$ and $\delta_{B,j}=1$ lead to either $\delta_{AB,j}=1$ or $\delta_{AB,j}=0$.
\end{enumerate}
We will refer to the case that $\delta_{A,j}=1$ and $\delta_{B,j}=1$ lead to $\delta_{AB,j}=1$ as an $\mathsf{OR}$ case (since $1\;\mathsf{OR}\;1=1$), and refer to the case that $\delta_{A,j}=1$ and $\delta_{B,j}=1$ lead to $\delta_{AB,j}=0$ as an $\mathsf{XOR}$ case (since $1\;\mathsf{XOR}\;1=0$). For convenience, the first three cases where $\delta_{A,j}\;\mathsf{OR}\;\delta_{B,j} = \delta_{A,j}\;\mathsf{XOR}\;\delta_{B,j}$ will be simply referred to as $\mathsf{OR}$ cases.


If there are only $\mathsf{OR}$ cases when combining difference vectors of multiple faults, a syndrome suitable for error correction can be easily found; whenever we find $\delta_i=0$ on the resulting difference vector, we know that no fault occurs on the $i$-th round so $\vec{s}_i$ can be used. In practice, however, the $\mathsf{OR}$ and $\mathsf{XOR}$ cases cannot be easily distinguished. Thus, finding that $\delta_i=0$ does not guarantee that there is no fault on the $i$-th round. 

For example, suppose that $t=3$ and the resulting difference vector is $\vec{\delta}=010010$. $\vec{\delta}$ can be from one of the following combinations of faults:
\begin{enumerate}
	\item $\mathrm{I}(1)$, $\mathrm{I}(2)$, and $\mathrm{II}(5)$ with difference vectors 100000, 110000, and 000010 where $\mathsf{XOR}$ cases happen when combining the 1st bits. In this case, $\vec{s}_1$ and $\vec{s}_2$ cannot be used for error correction.
	\item $\mathrm{I}(3)$, $\mathrm{I}(4)$, and $\mathrm{I}(5)$ with difference vectors 011000, 001100, and 000110 where $\mathsf{XOR}$ cases happen when combining the 3rd, 4th, and 5th bits. In this case, $\vec{s}_3$, $\vec{s}_4$, and $\vec{s}_5$ cannot be used for error correction.
	\item $\mathrm{II}(2)$, $\mathrm{I}(6)$, and $\mathrm{I}(7)$ with difference vectors 010000, 000011, and 000001 where $\mathsf{XOR}$ cases happen when combining the 6th bits. In this case, $\vec{s}_6$ and $\vec{s}_7$ cannot be used for error correction.
\end{enumerate}
In the example above, none of $\vec{s}_1$--$\vec{s}_7$ works for all cases, so error correction cannot be done accurately when $\vec{\delta}=010010$ is found.



Fortunately, an FTEC protocol is normally developed to handle a \emph{limited number of faults}. We can use this fact to determine whether a zero bit in the resulting difference vector can arise from the $\mathsf{XOR}$ case. For example, suppose that the total number of faults is limited to $t=3$ and the resulting difference vector is $\vec{\delta}=0100010$:
\begin{enumerate}
	\item $\vec{\delta}$ can be from combining $\mathrm{I}(1)$, $\mathrm{I}(2)$, and $\mathrm{II}(6)$ with difference vectors 1000000, 1100000, and 0000010 where $\mathsf{XOR}$ cases happen when combining the 1st bits. In this case, $\vec{s}_1$ and $\vec{s}_2$ cannot be used for error correction.
	\item Also, $\vec{\delta}$ can be from combining  $\mathrm{II}(2)$, $\mathrm{I}(7)$, and $\mathrm{I}(8)$ with difference vectors 0100000, 0000011, and 0000001 where $\mathsf{XOR}$ cases happen when combining the 7th bits. In this case, $\vec{s}_6$ and $\vec{s}_7$ cannot be used for error correction.	
	\item However, $\vec{\delta}$ \emph{cannot} be from combining  $\mathrm{I}(3)$, $\mathrm{I}(4)$, $\mathrm{I}(5)$, and $\mathrm{I}(6)$ with difference vectors 0110000, 0011000, 0001100, and 0000110 where $\mathsf{XOR}$ cases happen when combining the 3rd to the 6th bits since this case requires 4 faults in total.
	\item Note that $\vec{\delta}$ can be from combining  $\mathrm{II}(2)$, $\mathrm{II}(5)$, and $\mathrm{I}(6)$ with difference vectors 0100000, 0000100, and 0000110 where $\mathsf{XOR}$ cases happen when combining the 5th bits. Although a Type I fault occurs on the 6th round, $\vec{s}_6$ can still be used for error correction since $\vec{s}_6=\vec{s}_5$ and there is no Type I fault on the 5th round. In this case, any $\vec{s}_i$ can be used for error correction.
\end{enumerate}
From the example above, we know that at least one of the bits $\delta_3$--$\delta_5$ must be a bit zero arising from the $\mathsf{OR}$ case ($0 \;\mathsf{OR}\;0\;\mathsf{OR}\;0$) if the total number of faults is no more than 3. Moreover, $\vec{s}_3=\vec{s}_4=\vec{s}_5=\vec{s}_6$. Thus, $\vec{s}_3$--$\vec{s}_6$ correspond to the data error at the end of some round and all of them can be used for error correction.

We will develop a general algorithm for finding a syndrome for error correction for any $t$ using the ideas explained previously. For convenience, we will introduce the notions of $\mathsf{OR}$ and $\mathsf{XOR}$ zeros, and usable and unusable zero substrings as follows:

\begin{definition}{$\mathsf{OR}$ and $\mathsf{XOR}$ zeros}
	
	Let $\vec{\delta}$ be a difference vector obtained from combining the difference vectors of some faults, and suppose that some bit $\delta_i$ of $\vec{\delta}$ is zero. $\delta_i$ is said to be an \emph{$\mathsf{OR}$ zero} if it arises from the $\mathsf{OR}$ case of fault combination ($0\;\mathsf{OR}\;0=0$), and $\delta_i$ is said to be an \emph{$\mathsf{XOR}$ zero} if it arises from the $\mathsf{XOR}$ case of fault combination ($1\;\mathsf{XOR}\;1=0$).
	\label{def:real_zero}
\end{definition}

\begin{definition}{Usable and unusable zero substrings}
	
	Let a difference vector $\vec{\delta}$ be of the form $\vec{\delta}=\eta_1 1 \eta_2 1 \dots 1 \eta_c$ for some positive integer $c$, where $\eta_j$ ($j=1,\dots,c$) is a zero substring of the form $00\dots0$ (the length of $\eta_j$ can be zero). For any $\eta_j$ with positive length, if it is certain that $\eta_j$ contains at least one $\mathsf{OR}$ zero, then $\eta_j$ is said to be \emph{usable}; otherwise, $\eta_j$ is \emph{unusable}.
	\label{def:usable}
\end{definition}

For any $\vec{\delta}$ of the form $\eta_1 1 \eta_2 1 \dots 1 \eta_c$, we aim to find whether $\eta_j$ is usable for all $\eta_j$ with positive length. If a usable $\eta_j$ exists, we can use a syndrome of \emph{any} round corresponding to $\eta_j$ to do error correction; because at least one zero bit in $\eta_j$ is an $\mathsf{OR}$ zero and all rounds corresponding to the same $\eta_j$ give the same syndrome, the syndrome accurately represents the data error at the end of some round. Here we will use the fact that the total number of faults are limited to find a usable syndrome.


For each $\eta_j$ of length $\geq 1$ ($2 \leq j \leq c-1$), we can define $\alpha_j$ and $\beta_j$ to be the minimum numbers of faults that lead to the substrings $\eta_1 1 \dots 1 \eta_{j-1}$ and $\eta_{j+1} 1 \dots 1 \eta_{c}$ (the substrings before and after $1\eta_j1$). That is, suppose that,
\begin{equation}
	\vec{\delta}=\eta_1 1 \dots 1 \eta_{j-1} \overset{*}{\big|} 1 \eta_j 1 \overset{**}{\big|} \eta_{j+1} 1 \dots 1 \eta_c. \nonumber
\end{equation}
Then, $\alpha_j$ is equal to the total number of non-overlapping $11$ substrings plus the total number of remaining one bits before $*$, and $\beta_j$ is equal to the total number of non-overlapping $11$ substrings plus the total number of remaining one bits after $**$. For example, for a substring $\eta_j=000$ in $\vec{\delta}=1011000111101$, $\alpha_j=2$ and $\beta_j=3$.

For $\eta_1$ and $\eta_c$, we will define $\alpha_1=0$ and $\beta_c=0$, and define $\beta_1$ and $\alpha_c$ similarly to those of other $\eta_j$'s. If all bits of $\vec{\delta}$ are zeros, we can write $\vec{\delta}=\eta_1$ and define $\alpha_1=\beta_1=0$.

The following theorem states the sufficient and necessary condition for a zero substring to be usable.
\begin{theorem}
	Let $t$ be the maximum number of faults, $\vec{\delta}$ be a difference vector of the form $\eta_1 1 \eta_2 1 \dots 1 \eta_c$, and $\alpha_j$ and $\beta_j$ be the minimum numbers of faults leading to the substrings before and after $1 \eta_j 1$ (or $\eta_1 1$ or $1 \eta_c$). Suppose that the length of $\eta_j$ is $\gamma_j > 0$. Then, $\eta_j$ is usable if and only if $\alpha_j+\beta_j + \gamma_j \geq t$.
	\label{thm:condition_for_usable}
\end{theorem}

\begin{proof}	
	Consider $\eta_j$ where $j=2\dots,c-1$. The minimum number of faults that can cause the substring $1\eta_j 1$ with all $\mathsf{XOR}$ zeros is $\gamma_j+1$ (i.e., $1\eta_j 1 = 100\dots001 = 110\dots000 + 011\dots000 + \dots + 000\dots011$). In case that all zeros in $1\eta_j 1$ are $\mathsf{XOR}$ zeros, the minimum number of total faults that cause $\vec{\delta}$ must satisfy $\alpha_j+\beta_j+(\gamma_j+1) \leq t$ (or equivalently, $\alpha_j+\beta_j+\gamma_j < t$).
	
	If $\alpha_j+\beta_j + \gamma_j \geq t$, the zeros in $1\eta_j 1$ cannot be all $\mathsf{XOR}$ zeros (i.e., $1\eta_j 1$ must arise from less than $\gamma_j+1$ faults). In other words, there is at least one $\mathsf{OR}$ zero in $\eta_j$, so $\eta_j$ is usable. In contrast, if $\alpha_j+\beta_j+\gamma_j < t$, $\eta_j$ may arise from $\gamma_j+1$ faults. Because it is not certain whether there is an $\mathsf{OR}$ zero in $\eta_j$ or not, $\eta_j$ is unusable. 
	
	Similar analysis is applicable to $\eta_1$ and $\eta_c$, where the substrings $\eta_1 1$ and $1 \eta_c$ are considered instead of $1\eta_j 1$, and $\alpha_1=0$ and $\beta_c=0$ are defined. It is also applicable when all bits of $\vec{\delta}$ are zeros, in which $\vec{\delta}=\eta_1$ and $\alpha_1=\beta_1=0$.
\end{proof}

\cref{thm:condition_for_usable} can be interpreted as follows: Let $t$ be the total number of faults. For each $\eta_j$, $\alpha_j+\beta_j$ is the minimum number of occurred faults, thus $t-\alpha_j-\beta_j$ is the maximum number of remaining faults, while $\gamma_j+1$ is the number of rounds with repeated syndromes. If the syndromes are repeated more than the maximum number of remaining faults, it is certain that at least one round in the $\gamma_j+1$ rounds must have a correct syndrome which can be used for error correction.

By \cref{thm:condition_for_usable},	an algorithm for finding a usable zero substring from a given fault number $t_\mathrm{in}$ and a difference vector $\vec{\delta}_\mathrm{in}$ can be constructed as follows:
\begin{algorithm}	
	Let $t_\mathrm{in}$ be any number of faults and $\vec{\delta}_\mathrm{in}=\eta_1 1 \eta_2 1 \dots 1 \eta_c$ be a difference vector for some positive integer $c$. For each $\eta_j$ with length $\gamma_j>0$, calculate $\alpha_j$ and $\beta_j$ (as defined in \cref{thm:condition_for_usable}). If $\alpha_j+\beta_j+\gamma_j \geq t_\mathrm{in}$, return $\eta_j$ as a usable zero substring. 
	\label{alg:find_usable}
\end{algorithm}

Using \cref{alg:find_usable}, an FTEC protocol correcting up to $t$ faults that satisfies the strong FTEC conditions can be developed:
\begin{protocol}{FTEC protocol satisfying the strong FTEC conditions for a stabilizer code of any distance}
	
	Let $t=\lfloor (d-1)/2\rfloor$ be the weight of error that a stabilizer code of distance $d$ can correct. In each round of full syndrome measurements, measure stabilizer generators using the Shor syndrome extraction circuits. After the $i$-th round ($i\geq 2$), calculate $\delta_{i-1}$. Repeat syndrome measurements until one of the following conditions is satisfied, then perform error correction using the error syndrome corresponding to each condition:
	\begin{enumerate}
		\item If at least one usable $\eta_j$ is found by \cref{alg:find_usable} where $t_\mathrm{in}=t$ and $\vec{\delta}_\mathrm{in}=\vec{\delta}$ (the current difference vector), stop the syndrome measurements. Perform error correction using the syndrome corresponding to any zero in $\eta_j$.
		\item If the total number of non-overlapping $11$ substrings in $\vec{\delta}$ is $t$, stop the syndrome measurements. Perform error correction using the syndrome obtained from the latest round.
	\end{enumerate}
	\label{pro:strong_any_t} 
\end{protocol}
(The second condition to stop the syndrome measurements is introduced to count the number of occurred faults in case that all bits in $\vec{\delta}$ are ones and there is no zero substring of positive length.)

It is possible to find the number of rounds of syndrome measurements in the worst-case scenario of \cref{pro:strong_any_t} for any $t$. This number is the same as the minimum number of rounds required to guarantee that a usable syndrome exists in any case. The number can be found by the following theorem:
\begin{theorem}
	Let $t = \lfloor (d-1)/2\rfloor$, where $d \geq 3$ is the distance of a stabilizer code being used in \cref{pro:strong_any_t}. Performing the following number of rounds of full syndrome measurements is sufficient to guarantee that \cref{pro:strong_any_t} is strongly $t$-fault tolerant;
	\begin{enumerate}
		\item if $t$ is odd, performing $\left(\frac{t+3}{2}\right)^2-1$ rounds of full syndrome measurements is sufficient;
		\item if $t$ is even, performing $\left(\frac{t+2}{2}\right)\left(\frac{t+4}{2}\right)-1$ rounds of full syndrome measurements is sufficient.
	\end{enumerate}
	\label{thm:bound_strong} 
\end{theorem}

A proof of \cref{thm:bound_strong} is provided in \cref{sec:thm2_proof}.

Note that in some cases, we do not have to complete the full syndrome measurements in the very last round to find a usable syndrome. For example, if a generator measurement reveals that the first bit of $\vec{s}_{i}$ is different from the first bit of $\vec{s}_{i-1}$, we can immediately tell that $\delta_{i-1}$ is 1. To reduce the total number of generator measurements using this idea, we can modify \cref{pro:strong_any_t} by checking the difference between the syndromes of the two latest rounds ($i$ and $i-1$) more frequently, and running \cref{alg:find_usable} as soon as $\delta_{i-1}=1$ is found (or at the end of the $i$-th round if $\delta_{i-1}=0$).

Moreover, \cref{pro:strong_any_t} applied to a CSS code can be further optimized using the fact that $Z$-type and $X$-type errors can be corrected separately. Let $t=\lfloor(d-1)/2\rfloor$ be the number of faults that the protocol can correct (where $d$ is the code distance). We can first measure $X$-type generators repeatedly and a difference vector $\vec{\delta}_x$ will be obtained. A syndrome $\vec{s}_x$ suitable for $Z$-type error correction can be found using \cref{alg:find_usable} with $t_\mathrm{in}=t$ and $\vec{\delta}_\mathrm{in}=\vec{\delta}_x$. Once $\vec{s}_x$ is found, the minimum number of faults that occur during the $X$-type generator measurement, denoted by $t_\mathrm{oc}$, can be calculated by counting the total number of non-overlapping $11$ substrings plus the total number of remaining one bits in $\vec{\delta}_x$. After that, we can measure $Z$-type generators repeatedly, and a difference vector $\vec{\delta}_z$ can be obtained. A syndrome $\vec{s}_z$ suitable for $X$-type error correction can be found using \cref{alg:find_usable} with $t_\mathrm{in}=t-t_\mathrm{oc}$ and $\vec{\delta}_\mathrm{in}=\vec{\delta}_z$. Because the maximum number of remaining faults $t-t_\mathrm{oc}$ is used instead of $t$ in the latter part of the protocol, $\vec{s}_z$ could be found faster than $\vec{s}_x$, and the total number of generator measurements in the protocol could be reduced. 



{Recall that in the traditional Shor FTEC scheme where the Shor decoder is used, the repeated syndrome measurements are done until the syndromes are repeated $t+1$ times in a row; i.e., a substring $\eta_j$ in $\vec{\delta}$ with length $\gamma_j=t$ is found. With the notations introduced in this work, the Shor decoder can be considered as a special case where $\alpha_j$ and $\beta_j$ are defined to be 0 for any $\eta_j$. That is, for the Shor decoder, $\eta_j$ is usable iff $\gamma_j\geq t$ (by \cref{thm:condition_for_usable}).

One interesting aspect of our scheme is that the information from the past, the current, and the future rounds (which is contained in $\alpha_j$, $\gamma_j$, and $\beta_j$) is used to determine whether the syndrome of the current rounds (that leads to $\eta_j$) is suitable for error correction. Here, the word ``future'' refers to the fact that the syndromes obtained in the very first rounds can be found usable at a later stage of the protocol as more syndromes are collected. This is in contrast to the Shor decoder in which only information of the current rounds (the number of rounds with repeated syndromes) is used.

Note that although our adaptive scheme requires less time overhead compared to the traditional Shor scheme, it requires more classical processing since \cref{alg:find_usable} must be run after each round of syndrome measurements. Nevertheless, the classical time complexity of \cref{alg:find_usable} is $O(t^3)$, so the classical processing part is not likely to limit the performance of the adaptive scheme. The analysis of the classical time complexity of \cref{alg:find_usable} is provided in \cref{sec:complexity}.

\section{Adaptive measurements for Shor error correction satisfying the weak FTEC conditions}
\label{sec:weak_EC}

As previously mentioned in \cref{sec:Shor}, a code of high distance can be obtained in some code families without using code concatenation. In that case, the weak FTEC conditions in \cref{def:weak_FT} are sufficient to guarantee that fault tolerance can be achieved; it is not necessary to guarantee the weight of the output error when the weight of the input error is too high. In this section, we will develop an FTEC protocol similar to the protocol in \cref{sec:strong_EC}, but the weak FTEC conditions are considered instead of the strong FTEC conditions (the conditions in \cref{def:strong_FT}). The main goal of this section is to further reduce the number of rounds required to find a syndrome suitable for error correction for some families of codes in which the strong FTEC conditions need not be satisfied.

Similar to the EC scheme in \cref{sec:strong_EC} (and the traditional Shor scheme), stabilizer generators will be measured using the Shor syndrome extraction circuits. However, we will use a different idea to find a syndrome suitable for error correction. Let $r$ be the weight of the input error, $s$ be the number of faults in the protocol, and $t=\lfloor (d-1)/2\rfloor$ be the weight of error that a stabilizer code of distance $d$ can correct. To make sure that both ECCP and ECRP in \cref{def:weak_FT} are satisfied whenever $r+s \leq t$, we will use the following ideas when developing an FTEC protocol:  
\begin{enumerate}
	\item For any case with $r \geq 1$ and $r+s \leq t$, at least one syndrome suitable for error correction (a syndrome $\vec{s}_i$ obtained from some round $i$ which corresponds to the data error at the end of that round) must be found. In this case, error correction using $\vec{s}_i$ will remove the input error and the error from faults that occurred up to round $i$. Thus, the output state will be logically the same as the input state, and the weight of the output error (the error from faults after round $i$) will be $\leq s$.
	\item For any case with $r=0$ and $s\leq t$, if the difference vector of that case is the same as the difference vector of some case with $r \geq 1$ and $r+s \leq t$, at least one syndrome suitable for error correction must be found (the same syndrome must work for both cases since they cannot be distinguished by observing the difference vector). In this case, the error correction will remove the error from some faults that occurred early in the protocol, so the output state will be logically the same as the input state, and the weight of the output error will be $\leq s$.
	\item For any case with $r=0$ and $s\leq t$, if the difference vector of that case is different from the difference vectors of all cases with $r \geq 1$, a protocol can stop without doing any error correction. Because there is no input error, the output state will be logically the same as the input state, and the weight of the output error will be $\leq s$.
\end{enumerate} 

The syndrome of the first round $\vec{s}_1$ is special since it is related to the weight of the input error. If $\vec{s}_1\neq 0$, then either one of the following is true:
\begin{enumerate}[label=\Alph*)]
	\item the input error has weight $\geq 1$, and there are $s \leq t-1$ faults in the protocol (a Type I fault may or may not be present in the first round), or
	\item there is no input error, there are $s \leq t$ faults in the protocol, and the first round has a Type I fault ($I(1)$). 
\end{enumerate}
On the other hand, if $\vec{s}_1=0$, then either one of the following is true:
\begin{enumerate}[label=\Alph*)]\setcounter{enumi}{2}
	\item the input error has weight $\geq 1$, there are $s \leq t-1$ faults in the protocol, and the first round has a Type I fault ($I(1)$), or
	\item there is no input error, there are $s \leq t$ faults in the protocol, and the first round has no Type I faults (no $I(1)$).
\end{enumerate}
(Here we assume that there is at most one fault in each round, which is either Type I or Type II. Please see \cref{subsec:diff_multiple} for the validity of this assumption.)

To see how a syndrome suitable for error correction can be found from a difference vector $\vec{\delta}$ in each case of $\vec{s}_1$, let us consider an FTEC protocol correcting up to $t=1$ fault as an example. First, suppose that $r+s \leq 1$ and $\vec{s}_1 \neq 0$. There are two possibilities:
\begin{enumerate}
	\item $r=1$ and $s=0$ (Case A): In this case, $\delta_1 = 0$ ($\vec{s}_1=\vec{s}_2$). The error correction can be done using $\vec{s}_1$. 
	\item $r=0$ and $s=1$ (Case B): In this case, there must be $I(1)$ and no faults on the other rounds, so $\delta_1 = 1$ ($\vec{s}_1\neq\vec{s}_2$). The error correction is not necessary in this case.
\end{enumerate}
We can see that when $\vec{s}_1 \neq 0$, the protocol will find a syndrome suitable for error correction or stop without doing error correction within 2 rounds. Next, suppose that $\vec{s}_1=0$. There is only one possible case: there is no input error and the first round has no Type I fault (Case D). Thus, we can stop without doing error correction whenever $\vec{s}_1 = 0$. Only 1 round of full syndrome measurements is needed in this case. Using these ideas, an FTEC protocol for $t=1$ can be constructed as follows:
\begin{protocol}{FTEC protocol satisfying the weak FTEC conditions for a stabilizer code of distance 3}
	
	In each round of full syndrome measurements, measure stabilizer generators using the Shor syndrome extraction circuits.
	After the syndrome $\vec{s}_1$ from the first round is obtained, do the following:
	\begin{enumerate}
		\item If $\vec{s}_1 \neq 0$, repeat the syndrome measurements to obtain $\vec{s}_2$.
		\begin{enumerate}
			\item If $\vec{s}_1=\vec{s}_2$, perform error correction using $\vec{s}_1$.
			\item If $\vec{s}_1\neq\vec{s}_2$, stop and do nothing.
		\end{enumerate}
		\item If $\vec{s}_1 = 0$, stop and do nothing.
	\end{enumerate} 
	\label{pro:weak_t_1} 
\end{protocol}
In fact, \cref{pro:weak_t_1} is similar to the protocol for a stabilizer code of distance 3 proposed by Delfosse and Reichart in \cite{DR20}. 

In case that $t\geq 2$, finding a usable zero substring (which gives a syndrome suitable for error correction) from a difference vector obtained from repeated syndrome measurements can be more complicated. Suppose that a difference vector is $\vec{\delta}=\delta_1 \delta_2 \dots \delta_m$ for some positive integer $m$. We will use the following procedures to find a usable zero substring $\eta_j$:
\begin{enumerate}
	\item If $\vec{s}_1 \neq 0$, a usable zero substring will be found by \cref{alg:find_usable} with $t_\mathrm{in}=t-1$ and $\vec{\delta}_\mathrm{in}=\delta_2 \dots \delta_m$ (the first bit of $\vec{\delta}$ is removed).
	\item If $\vec{s}_1 = 0$, a usable zero substring will be found by \cref{alg:find_usable} with $t_\mathrm{in}=t$ and $\vec{\delta}_\mathrm{in}=0 \delta_1 \delta_2 \dots \delta_m$ (bit zero is added to the beginning of $\vec{\delta}$).
\end{enumerate}	

To see how these work, assume that $r+s \leq t$ and first consider the case that $\vec{s}_1\neq 0$. Case A has $r \geq 1$ and $s \leq t-1$, while Case B has $r=0$, $s \leq t$, and $I(1)$. In both cases, the second rounds onward have $\leq t-1$ faults. When a usable zero substring $\eta_j$ is found by \cref{alg:find_usable} with $t_\mathrm{in}=t-1$ and $\vec{\delta}_\mathrm{in}=\delta_2 \dots \delta_m$, it is guarantee that there is at least one $\mathsf{OR}$ zero in $\eta_j$. Thus, the syndrome corresponding to $\eta_j$ can be used for error correction in both cases.

Next, consider the case that $\vec{s}_1 = 0$. Case C has $r \geq 1$, $s \leq t-1$, and $I(1)$, while Case D has $r=0$, $s \leq t$, and no $I(1)$. Suppose that $\vec{\delta} = \eta_1 1 \eta_2 1 \dots 1 \eta_c$, and we run \cref{alg:find_usable} with $t_\mathrm{in}=t$ and $\vec{\delta}_\mathrm{in}=0 \eta_1 1 \eta_2 1 \dots 1 \eta_c$. 
\begin{enumerate}
	\item If the algorithm finds that $\eta_j$ with $j\geq 2$ is usable, the syndrome corresponding to $\eta_j$ can be used for error correction in both cases since there is an $\mathsf{OR}$ zero in $\eta_j$.
	\item Suppose that the algorithm finds that $0\eta_1$ is usable. By \cref{thm:condition_for_usable}, this can happen only when $\beta_1+\gamma_1+1\geq t$ (where $\gamma_1$ is the number of zeros in $\eta_1$). In Case C where $I(1)$ is present, the substring $\eta_1 1$ must arise from $\gamma_1+1$ faults. Because the total number of faults in Case C must satisfy $s\leq t-1$, i.e., $\beta_1+\gamma_1+1 \leq t-1$ must hold, the algorithm will never find that $0\eta_1$ is usable when Case C happens. In other words, the algorithm will find that $0\eta_1$ is usable only when Case D happens. It is okay not to perform any error correction in this case since Case D has no input error (the syndrome corresponding to $0 \eta_1$ is $\vec{s}_1 = 0$ which leads to no error correction).
\end{enumerate}

Here we can construct an FTEC protocol correcting up to $t$ faults that satisfies the weak FTEC conditions as follows:
\begin{protocol}{FTEC protocol satisfying the weak FTEC conditions for a stabilizer code of distance $d \geq 5$}
	
	Let $t=\lfloor (d-1)/2\rfloor$ be the weight of error that a stabilizer code of distance $d$ can correct ($d \geq 5$). In each round of full syndrome measurements, measure stabilizer generators using the Shor syndrome extraction circuits. After the $i$-th round ($i\geq 2$), calculate $\delta_{i-1}$. Repeat syndrome measurements until one of the following conditions is satisfied, then perform error correction using the error syndrome corresponding to each condition:
	\begin{enumerate}
		\item If $\vec{s}_1 \neq 0$, obtain $\vec{\delta}'$ by removing the first bit of $\vec{\delta}$.
		\begin{enumerate}
			\item If at least one usable $\eta_j$ is found by \cref{alg:find_usable} where $t_\mathrm{in}=t-1$ and $\vec{\delta}_\mathrm{in}=\vec{\delta}'$, stop the syndrome measurements. Perform error correction using the syndrome corresponding to any zero in $\eta_j$.
			\item If the total number of non-overlapping $11$ substrings in $\vec{\delta}'$ is $t-1$, stop the syndrome measurements. Perform error correction using the syndrome obtained from the latest round.
		\end{enumerate}
		\item If $\vec{s}_1 = 0$, obtain $\vec{\delta}'$ by adding bit zero to the beginning of $\vec{\delta}$.
		\begin{enumerate}
			\item If at least one usable $\eta_j$ is found by \cref{alg:find_usable} where $t_\mathrm{in}=t$ and $\vec{\delta}_\mathrm{in}=\vec{\delta}'$, stop the syndrome measurements. Perform error correction using the syndrome corresponding to any zero in $\eta_j$. 
			\item If the total number of non-overlapping $11$ substrings in $\vec{\delta}'$ is $t$, stop the syndrome measurements. Perform error correction using the syndrome obtained from the latest round.
		\end{enumerate}	
	\end{enumerate}
	\label{pro:weak_any_t} 
\end{protocol}

The number of rounds of syndrome measurements in the worst-case scenario of \cref{pro:weak_any_t} for each of $\vec{s}_1 \neq 0$ and $\vec{s}_1 = 0$ cases, which is also the minimum number of rounds required to guarantee that a usable syndrome exists in each case, can be found by the following theorem:
\begin{theorem}
	Let $t = \lfloor (d-1)/2\rfloor$, where $d \geq 5$ is the distance of a stabilizer code being used in \cref{pro:weak_any_t}. Performing the following number of rounds of full syndrome measurements is sufficient to guarantee that \cref{pro:weak_any_t} is weakly $t$-fault tolerant;
	\begin{enumerate}
		\item if $\vec{s}_1 \neq 0$,
		\begin{enumerate}
			\item if $t$ is even, performing $\left(\frac{t+2}{2}\right)^2$ rounds of full syndrome measurements is sufficient;
			\item if $t$ is odd, performing $\left(\frac{t+1}{2}\right)\left(\frac{t+3}{2}\right)$ rounds of full syndrome measurements is sufficient;
		\end{enumerate}
		\item if $\vec{s}_1 = 0$,
		\begin{enumerate}
			\item if $t$ is even, performing $\left(\frac{t+2}{2}\right)\left(\frac{t+4}{2}\right)-2$ rounds of full syndrome measurements is sufficient;
			\item if $t$ is odd, performing $\left(\frac{t+3}{2}\right)^2-2$ rounds of full syndrome measurements is sufficient.
		\end{enumerate}
	\end{enumerate}	
	\label{thm:bound_weak} 
\end{theorem}

\begin{proof}
	Suppose that \cref{alg:find_usable} is run with $t_\mathrm{in}$ and $\vec{\delta}_\mathrm{in}$. From the proof of \cref{thm:bound_strong} (provided in \cref{sec:thm2_proof}), the maximum length of $\vec{\delta}_\mathrm{in}$ such that no usable zero substring is found by \cref{alg:find_usable} is $\left(\frac{t_\mathrm{in}+3}{2}\right)^2-3$ when $t_\mathrm{in}$ is odd (or $\left(\frac{t_\mathrm{in}+2}{2}\right)\left(\frac{t_\mathrm{in}+4}{2}\right)-3$ when $t_\mathrm{in}$ is even). Let $\vec{\delta}=\delta_1\delta_2\dots\delta_m$ be the difference vector obtained from repeated syndrome measurements. First, consider the case that $\vec{s}\neq 0$ in which $t_\mathrm{in}=t-1$ and $\vec{\delta}_\mathrm{in}=\vec{\delta}'=\delta_2\dots\delta_m$ ($\mathrm{length}(\vec{\delta}')=\mathrm{length}(\vec{\delta})-1$). if $t$ is even (or $t$ is odd), the maximum length of $\vec{\delta}'$ with no usable zero substring is $\left(\frac{t+2}{2}\right)^2-3$ (or $\left(\frac{t+1}{2}\right)\left(\frac{t+3}{2}\right)-3$). That is, a usable zero substring exists in $\vec{\delta}'$ when the length of $\vec{\delta}$ is $\left(\frac{t+2}{2}\right)^2-1$ (or $\left(\frac{t+1}{2}\right)\left(\frac{t+3}{2}\right)-1$). This is always achievable when $\left(\frac{t+2}{2}\right)^2$ (or $\left(\frac{t+1}{2}\right)\left(\frac{t+3}{2}\right)$) rounds of full syndrome measurements are performed.
	
	Next, consider the case that $\vec{s}= 0$ in which \cref{alg:find_usable} is run with $t_\mathrm{in}=t$ and $\vec{\delta}_\mathrm{in}=\vec{\delta}'=0\delta_1\delta_2\dots\delta_m$ ($\mathrm{length}(\vec{\delta}')=\mathrm{length}(\vec{\delta})+1$). If $t$ is even (or $t$ is odd), the maximum length of $\vec{\delta}'$ with no usable zero substring is $\left(\frac{t+2}{2}\right)\left(\frac{t+4}{2}\right)-3$ (or $\left(\frac{t+3}{2}\right)^2-3$). That is, a usable zero substring exists in $\vec{\delta}'$ when the length of $\vec{\delta}$ is $\left(\frac{t+2}{2}\right)\left(\frac{t+4}{2}\right)-3$ (or $\left(\frac{t+3}{2}\right)^2-3$). This is always achievable by performing $\left(\frac{t+2}{2}\right)\left(\frac{t+4}{2}\right)-2$ (or $\left(\frac{t+3}{2}\right)^2-2$) rounds of full syndrome measurements.
	
	In any case, the syndrome corresponding to a usable zero substring found by \cref{alg:find_usable} can be used to perform error correction as described earlier in this section.
\end{proof}

Similar to the protocols in \cref{sec:strong_EC}, we do not have to complete the full syndrome measurements in the very last round of \cref{pro:weak_t_1} or \cref{pro:weak_any_t}; we can check the difference between the syndromes of the two latest rounds ($i$ and $i-1$) more frequently, and run \cref{alg:find_usable} to find a syndrome suitable for error correction or stop as soon as we are certain that $\delta_{i-1}=1$ (the syndromes $\vec{s}_{i}$ and $\vec{s}_{i-1}$ differ by at least one bit). Also, for a CSS code in which $X$-type and $Z$-type errors can be corrected separately, we can further reduce the total number of stabilizer generators; after a syndrome for correcting errors of $X$-type or $Z$-type is found, the minimum number of occurred faults can be calculated and used to find a syndrome for correcting errors of another type. See the technique proposed in \cref{sec:strong_EC} for more details.

\section{Decoder Comparison}
\label{sec:comparison}




Compared to the traditional Shor FTEC scheme, our FTEC schemes with adaptive decoders require fewer rounds of syndrome measurements in the worst-case scenario. Because the fault-tolerant threshold is related to the number of fault combinations that can cause a logical error (which is related to the total number of gates in the whole protocol), we expect to see a higher threshold when the number of rounds is reduced. In this section, we compare our adaptive decoders with the Shor decoder both analytically and numerically.

\subsection{Improvement of the lower bound of the fault-tolerant threshold for a concatenated code}
\label{subsec:bound}

Consider an FTEC protocol that uses a stabilizer code of distance $d=2t+1$ with Shor syndrome extraction circuits. Let $L$ be the number of locations (a single qubit preparation, a 1-qubit or 2-qubit gate, or a single qubit measurement) in each round of full syndrome measurements, and suppose that each location can fail with probability $p$. Also, let $r_1$ and $r_2$ be the numbers of rounds required in the worst-case scenario for the protocols with the traditional Shor decoder and the adaptive strong decoder from \cref{sec:strong_EC}. Because both protocols are strongly $t$-fault tolerant, error correction can fail only when there are at least $t+1$ faults in each protocol.

Let us first consider the protocol with the traditional Shor decoder. Pessimistically assuming that the protocol fails every time when $t+1$ faults occur, the number of fault combinations that can cause the protocol to fail is $\binom{r_1L}{t+1}$. The logical error rate $p^{(1)}$ when the quantum data is encoded once satisfies,
\begin{equation}
	p^{(1)}\leq \binom{r_1L}{t+1} p^{t+1}. \label{eq:th_lv1}
\end{equation}
Let $p_{T,1} = \binom{r_1L}{t+1}^{-1/t}$. \cref{eq:th_lv1} can be rewritten as,
\begin{equation}
	\frac{p^{(1)}}{p_{T,1}}\leq \left(\frac{p}{p_{T,1}}\right)^{t+1}.
\end{equation}

Suppose that the FTEC protocol for a concatenated code is constructed by recursively replacing each qubit in the protocol by a block of code and replacing each physical gate by a logical gate, the logical error rate $p^{(m)}$ when the quantum data is encoded $m$ times satisfies,
\begin{equation}
	\frac{p^{(m)}}{p_{T,1}}\leq \left(\frac{p^{(m-1)}}{p_{T,1}}\right)^{t+1} \leq \left(\frac{p}{p_{T,1}}\right)^{(t+1)^m}.
\end{equation}
The logical error rate $p^{(m)}$ can be suppressed to an arbitrarily small value by increasing the level of concatenation whenever $p \leq p_{T,1}$; that is, $p_{T,1}$ is the fault-tolerant threshold for a concatenated code when the traditional Shor scheme is used. However, in practice not all combinations of $t+1$ faults cause the protocol to fail, so $p_{T,1} = \binom{r_1L}{t+1}^{-1/t}$ is actually a \emph{lower bound} of the fault-tolerant threshold for the traditional Shor scheme.

Using similar analysis, a lower bound of the fault-tolerant threshold for the adaptive strong decoder is $p_{T,2} = \binom{r_2L}{t+1}^{-1/t}$. We find that
\begin{equation}
	\left(\frac{p_{T,2}}{p_{T,1}}\right)^t = \frac{\binom{r_1L}{t+1}}{\binom{r_2L}{t+1}} = \frac{(r_1L)(r_1L-1)\cdots(r_1L-t)}{(r_2L)(r_2L-1)\cdots(r_2L-t)}.
\end{equation}
Since $(r_1L-x)/(r_2L-x) \geq (r_1L)/(r_2L)$ for any $x \in \{0,\dots,t\}$ when $r_1 \geq r_2$, the following holds;
\begin{equation}
	\left(\frac{p_{T,2}}{p_{T,1}}\right)^t \geq \left(\frac{r_1}{r_2}\right)^{t+1}.
\end{equation}
That is,
\begin{equation}
	p_{T,2} \geq p_{T,1}\left(\frac{r_1}{r_2}\right)^{1+\frac{1}{t}}. \label{eq:improved_bound}
\end{equation}
This implies that the lower bound of the fault-tolerant threshold for a concatenated code can be improved when the adaptive strong decoder is used instead of the Shor scheme. The improvement factor becomes larger as the distance of the base code for concatenation increases, and it reaches a factor of $4$ as $t\rightarrow \infty$.

\begin{figure*}[tbp]
	\centering
	\includegraphics[width=0.7\textwidth]{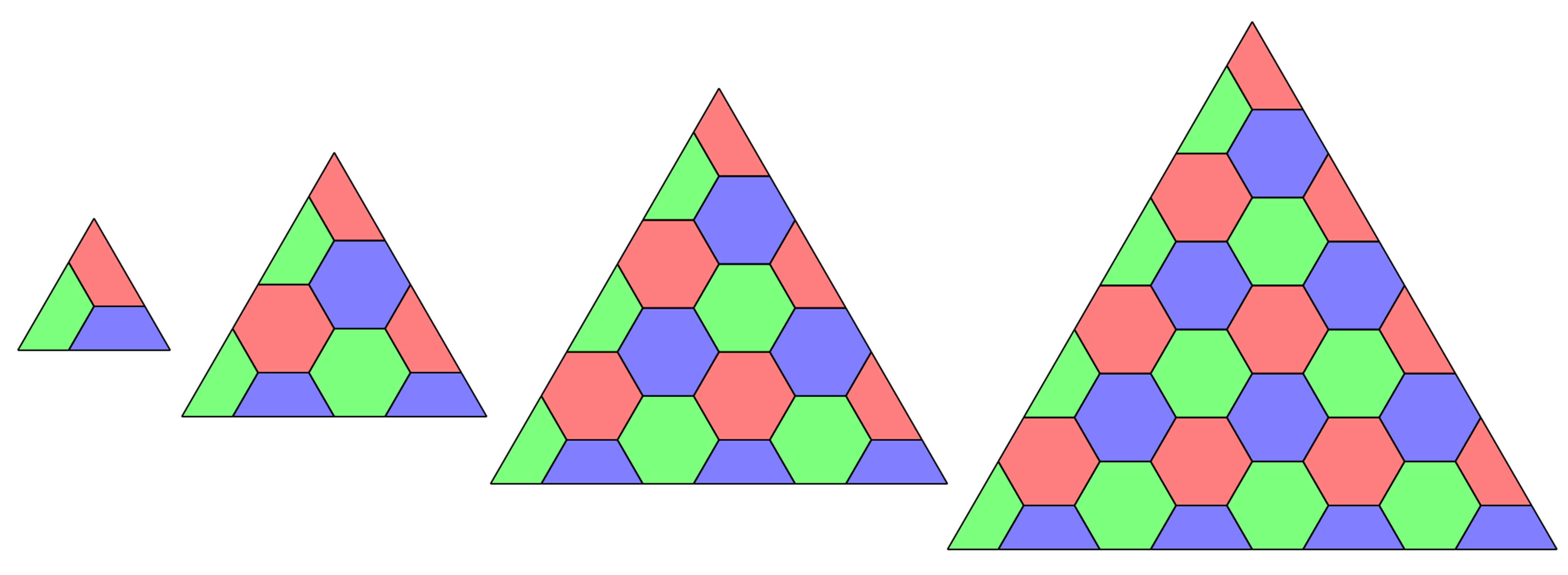}
	\caption{Hexagonal color codes of distance 3, 5, 7, and 9.}
	\label{fig:CC_d3-9}
\end{figure*}

It should be noted that the improvement factor of the actual fault-tolerant threshold can be lower than the factor in \cref{eq:improved_bound} since the ratio of the average numbers of rounds is most likely lower than the ratio of the numbers of rounds in the worst-case scenarios. Also, the analysis above is not applicable to families of codes in which a code of high distance can be obtained without concatenation, such as topological codes. In the next section, we present numerical simulations of hexagonal color codes which show that our adaptive strong and adaptive weak decoders can improve the pseudothreshold and the average number of rounds.


\subsection{Improvement of the pseudothreshold and the average number of rounds for hexagonal color codes}
\label{subsec:num_threshold}

To verify that our adaptive schemes are fault tolerant and have some advantages over the traditional Shor scheme, we simulate FTEC protocols on the hexagonal color codes of distance $3,5,7,$ and $9$. The hexagonal color code of distance $d$ is a \codepar{(3d^2+1)/4,1,d} code \cite{BM06}. It is a CSS code \cite{CS96,Steane96b} in which $X$-type and $Z$-type generators have the same form. Hexagonal color codes of distance $3,5,7,$ and $9$ are illustrated in \cref{fig:CC_d3-9}.

Our simulations consider the circuit-level depolarizing noise model as described below. 
\begin{enumerate}
	\item Each 1-qubit gate is followed by $X$, $Y$, or $Z$ error with probability $p/3$ each, i.e., the 1-qubit symmetric depolarizing noise with error probability $p$.
	\item Each 2-qubit gate is followed by a 2-qubit Pauli error of the form $P_1\otimes P_2$ (where $P_1,P_2 \in \{I,X,Y,Z\}$ and $P_1\otimes P_2 \neq I\otimes I$) with probability $p/15$ each.
	\item Ancilla qubits for a measurement of weight-$w$ generator are initially prepared in a cat state of the form $\left(\left|0\right\rangle^{\otimes w}+\left|1\right\rangle^{\otimes w}\right)/\sqrt{2}$ and each qubit is subjected to the 1-qubit symmetric depolarizing noise with error probability $p$.
	\item After each qubit measurement, the classical outcome is subjected to a bit-flip error with probability $p$.
	\item There is no idling (or wait-time) error.
\end{enumerate}

Each stabilizer generator is measured using the Shor syndrome extraction circuit. After each round of full syndrome measurements, the difference vector is calculated. Syndrome measurements are performed until the difference vector satisfies the stopping condition of the decoder being used (which is the Shor, the adaptive strong, or the adaptive weak decoder). Once a usable syndrome is found, an EC operator is obtained by the minimum-weight decoder for a hexagonal color code, and it is applied to the output error from the last round of syndrome measurements. Finally, an ideal error correction is applied, and we verify whether the output error is a logical error.

(In our implementation of the Shor decoder, syndrome measurements stop when either (1) the syndromes are repeated $t+1$ times in a row, or (2) the total number of rounds reaches $(t+1)^2$. Then, we use the syndrome from the last round for error correction. Note that the case in which the first condition is not satisfied but the number of rounds exceeds $(t+1)^2$ only happens when the total number of faults is more than $t$. The second condition is introduced to prevent a simulation blowup.)

\begin{table*}[tbp]
	\begin{center}
		\begin{tabular}{| c | c | c | c | c | c | c | c |}
			\hline
			\multirow{2}{*}{Distance} & \multirow{2}{*}{Decoder} & Pseudothreshold & \multicolumn{5}{|c|}{Average number of rounds at error rate $p$} \\
			\cline{4-8}
			& & ($\times 10^{-4}$) & $p=10^{-4}$ & $p=10^{-3}$ & $p=10^{-2}$ & $p=10^{-1}$ & $\;\;\;p=1\;\;\;$ \\
			\hline 
			\multirow{3}{*}{$d=3$} & Shor & $4.12 \pm 0.90$ & $2.02$ & $2.17$ & $3.24$ & $3.95$ & $3.96$ \\
			\cline{2-8}
			& Strong & $3.88 \pm 0.74$ & $2.01$ & $2.12$ & $2.69$ & $2.98$ & $2.98$ \\
			\cline{2-8}
			&Weak & $16.4 \pm 5.0$ & $1.01$ & $1.06$ & $1.45$ & $1.98$ & $1.99$ \\
			\hline
			\multirow{3}{*}{$d=5$} & Shor & $3.28 \pm 0.25$ & $3.13$ & $4.55$ & $8.98$ & $9.00$ & $9.00$ \\
			\cline{2-8}
			& Strong & $4.25 \pm 0.48$ & $3.07$ & $3.64$ & $4.96$ & $5.00$ & $5.00$ \\
			\cline{2-8}
			&Weak & $5.48 \pm 1.12$ & $2.05$ & $2.47$ & $3.86$ & $4.00$ & $4.00$ \\
			\hline
			\multirow{3}{*}{$d=7$} & Shor & $1.96 \pm 0.07$ & $4.52$ & $11.05$ & $16.00$ & $16.00$ & $16.00$ \\
			\cline{2-8}
			& Strong & $3.59 \pm 0.25$ & $4.26$ & $5.78$ & $7.00$ & $7.00$ & $7.00$ \\
			\cline{2-8}
			&Weak & $4.05 \pm 0.35$ & $3.22$ & $4.56$ & $5.99$ & $6.00$ & $6.00$ \\
			\hline
			\multirow{3}{*}{$d=9$} & Shor & $1.19 \pm 0.01$ & $6.50$ & $23.45$ & $25.00$ & $25.00$ & $25.00$ \\
			\cline{2-8}
			& Strong & $2.75 \pm 0.10$ & $5.69$ & $8.22$ & $9.00$ & $9.00$ & $9.00$ \\
			\cline{2-8}
			&Weak & $2.99 \pm 0.12$ & $4.63$ & $7.03$ & $8.00$ & $8.00$ & $8.00$ \\
			\hline
		\end{tabular}
	\end{center}
	\caption{Pseudothresholds and the average number of rounds for the hexagonal color codes of distance $d=3,5,7,9$ when the Shor, the adaptive strong, or the adaptive weak decoder is applied.}
	\label{tab:sim_results}
\end{table*}

In our numerical simulations, we generate the syndrome extraction circuits using Cirq \cite{Cirq22} and use Stim's Pauli frame simulator \cite{Gidney21} as a C++ library to sample from them. On our plots, the data points are calculated from between $5\times 10^8$ and $10^5$ samples at lower error rates, and at high error rates, we stop the sampling when 1000 out of all samples resulted in a logical error. The plots are generated with Sinter \cite{Gidney21}. Each data point is the number of logical errors divided by the number of samples. 
As the experiments follow the binomial distribution, this ratio is an estimator of the logical error rate $p_L$. 
Using Bayesian hypothesis testing, Sinter \cite{Gidney21} calculates the maximum and the minimum possible values of $p_L$, and the bounds are represented with the shaded area. Any other distribution with $p_L$ outside these limits is highly unlikely, quantified by the Bayes factor of $10^3$. The software developed for the simulations in this work is similar to the one presented in \cite{PTHB23}, where the technical and implementation details of similar decoders are discussed.

\begin{figure*}[tbp]
	\centering
	\includegraphics[width=0.9\textwidth]{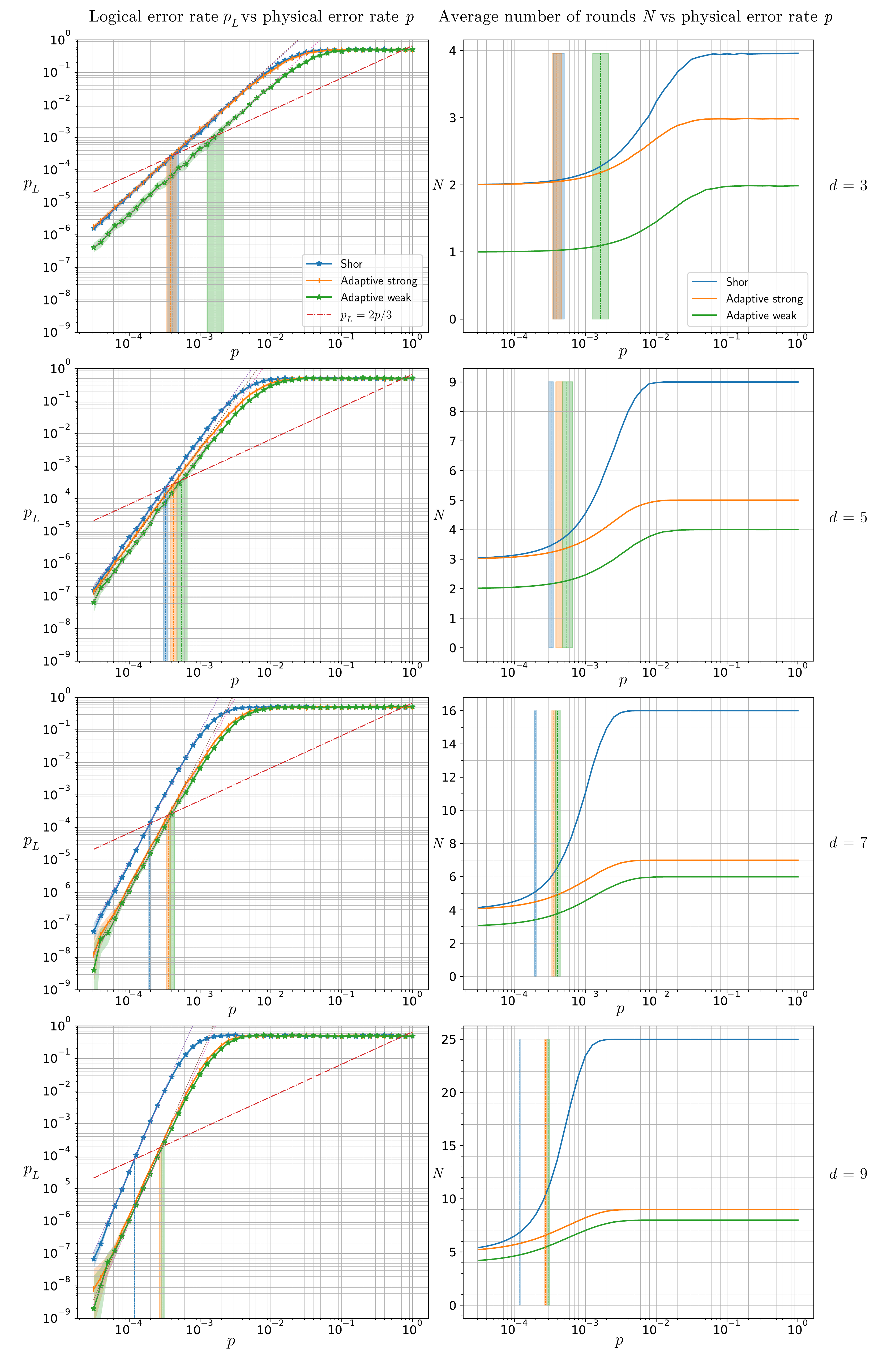}
	\caption{The logical error rate and the average number of rounds at each physical error rate for the hexagonal color codes of distance $d=3,5,7,9$ when the Shor, the adaptive strong, or the adaptive weak decoder is applied. The vertical lines indicate the pseudothresholds.}
	\label{fig:plots}
\end{figure*}

Our numerical results are shown in \cref{fig:plots} and \cref{tab:sim_results}.
From the plots of logical error rates $p_L$ versus physical error rates $p$ for each hexagonal color code (\cref{fig:plots} left), we observe that all plots for the code of distance $d=2t+1$ are parallel to $p_L=p^{t+1}$ at low error rates. This means that the Shor, the adaptive strong, and the adaptive weak decoders are $t$-fault tolerant, i.e., the code distance is preserved. 
Furthermore, both adaptive strong and adaptive weak decoders can increase the pseudothreshold\footnote{The pseudothreshold is the physical error rate in which the $p_L(p)$ curve intersects with the $p_L=2p/3$ line.} of each code (except for $d=3$ in which the performances of the Shor and the adaptive strong decoders are very similar). 
As the code distance grows, the difference in the performances between the Shor and the adaptive strong decoders becomes larger. In contrast, the difference between the adaptive strong and the adaptive weak decoders becomes smaller.

From the plots of average number of rounds $N$ versus physical error rates $p$ for each hexagonal color code (\cref{fig:plots} right), we find that the average number of rounds of the Shor and the adaptive strong decoders are similar when $p$ is much lower than the pseudothresholds, but the difference becomes more noticeable as $p$ increases. We also observe that the average number of rounds of the adaptive weak decoder are roughly 1 round fewer than that of the adaptive strong decoder for the entire range of $p$.


Moreover, we observe an interesting behavior of our adaptive decoders at the high physical error rate regime. For the hexagonal color code of distance $d=2t+1$, the average number of rounds for the adaptive strong and the adaptive weak decoders approach $2t+1$ and $2t$, respectively. In contrast, the number for the Shor decoder approaches $(t+1)^2$. This is because the chance of observing repeated syndromes becomes exponentially small at high error rates. Thus, an average is dominated by the case that all bits of the difference vector are ones, which causes the Shor, the strong, and the weak decoders to stop at $(t+1)^2$, $2t+1$, and $2t$ rounds, respectively. Since $N$ monotonically increases as $p$ increases in every case, the observation suggests that FTEC on a code of distance $d$ with our adaptive strong (or weak) decoder can be achieved by repeating no more than $d$ (or $d-1$) rounds of syndrome measurements on average. We also expect this result to hold for any stabilizer code.

\section{Discussions and conclusions}
\label{sec:discussion}

\begin{table*}[tbp]
	\begin{center}
		\begin{tabular}{| c | c | c | c | c | c | c | c | c | c | c |}
			\hline
			\multicolumn{2}{|c|}{\multirow{2}{*}{Protocol}} & \multicolumn{9}{|c|}{The maximum number of rounds} \\ \cline{3-11}
			\multicolumn{2}{|c|}{} & $t=1$ & $t=2$ & $t=3$ & $t=4$ & $t=5$ & $t=6$ & $t=7$ & $t=8$ & $t=9$ \\
			\hline
			\multicolumn{2}{|c|}{Strong FTEC} & 3 & 5 & 8 & 11 & 15 & 19 & 24 & 29 & 35  \\
			\hline
			\multirow{2}{*}{Weak FTEC} & if $\vec{s}_1\neq 0$ & 2 & 4 & 6 & 9 & 12 & 16 & 20 & 25 & 30 \\ \cline{2-11}
			& if $\vec{s}_1=0$ & 1 & 4 & 7 & 10 & 14 & 18 & 23 & 28 & 34 \\
			\hline
			\multicolumn{2}{|c|}{Traditional Shor \cite{Shor96}} & 4 & 9 & 16 & 25 & 36 & 49 & 64 & 81 & 100  \\
			\hline
		\end{tabular}
	\end{center}
	\caption{The maximum numbers of rounds of full syndrome measurements required for FTEC protocols satisfying the strong FTEC conditions (\cref{def:strong_FT}), the weak FTEC conditions (\cref{def:weak_FT}), and the traditional Shor FTEC protocol when applying to a stabilizer code that can correct up to $t=1,\dots,9$ errors.}
	\label{tab:results_t1-10}
\end{table*}

In this work, we present FTEC schemes which are improved versions of the Shor FTEC scheme. Our protocols measure the error syndromes repeatedly using the Shor extraction circuits, then the difference vector is calculated from the differences of syndromes between any two consecutive rounds. Afterwards, a syndrome for error correction is determined by the pattern of the difference vector. We call this kind of technique an adaptive measurement technique since the condition to stop repeated syndrome measurements changes dynamically depending on the measurement results. This is in contrast to the traditional Shor FTEC scheme, where syndrome measurements are done until the results are repeated $t+1$ times in a row regardless of the measurement results. 

Our protocols with the adaptive strong and the adaptive weak decoders (which satisfy the strong and the weak FTEC conditions in \cref{def:strong_FT,def:weak_FT}) are developed in \cref{sec:strong_EC,sec:weak_EC}, respectively. The protocol with the adaptive strong decoder can be applied to a concatenated code by replacing each qubit (and each physical gate) in the protocol with a code block (and the corresponding logical gate). On the other hand, the protocol with the adaptive weak decoder does not support code concatenation but can be applied to code families in which a code of higher distance can be obtained without code concatenation. The maximum number of rounds required for the protocols with the strong and the weak decoders are proved in \cref{thm:bound_strong,thm:bound_weak}. For a stabilizer code that can correct up to $t$ errors, the protocol with the strong decoder requires no more than $(t+3)^2/4-1$ rounds of syndrome measurements, while the protocol with the weak decoder requires no more than $(t+3)^2/4-2$ rounds (or no more than $(t+2)^2/4$ rounds) if the syndrome from the first round is trivial (or nontrivial). The maximum numbers of required rounds for $t=1,\dots,9$ for each of our protocols are compared in \cref{tab:results_t1-10}.


In \cref{subsec:bound}, we analytically show that compared to the Shor decoder, our adaptive strong decoder can improve the lower bound of the fault-tolerant threshold when applied to a concatenated code. The improvement factor for the bound approaches $4$ times when the distance of the base code for concatenation approaches infinity. However, it should be noted that the number of rounds used in our analysis is the number of rounds for the worst-case scenario, not the average number of rounds. Thus, the improvement factor for the actual threshold might be lower than $4$ times as the worst-case scenario is very unlikely. In \cref{subsec:num_threshold}, we present numerical simulations of FTEC protocols with the Shor, the adaptive strong, and the adaptive weak decoders on the hexagonal color codes of distance $d=3,5,7,9$. The results show that the protocols with adaptive decoders preserve the code distance and can indeed improve the pseudothreshold. Interestingly, we also find that while the maximum number of rounds for each of our adaptive decoders grows quadratically as the code distance grows, the maximum of the average number of rounds grows linearly (in contrast to the Shor decoder in which both numbers grow quadratically). The results suggest that our FTEC protocols with adaptive decoders applied to any stabilizer code of distance $d$ require no more than $d$ rounds of syndrome measurements on average.


We expect that the adaptive measurement technique from \cref{sec:strong_EC,sec:weak_EC} could be applicable to other fault-tolerant protocols in which eigenvalues of Pauli operators are measured using circuits similar to the Shor syndrome extraction circuit, and the ancilla measurement faults are handled by repeated measurements. An extension of our adaptive decoders to flag FTEC is developed in \cite{PTHB23}. Other possible protocols that could be improved by our adaptive scheme are fault-tolerant protocols for operator measurements \cite{NC00}, and for code-switching \cite{PR13,ADP14,Bombin15,KB15}. However, one has to make sure that the protocols are modified in a way that the conditions for fault tolerance for such tasks are satisfied. A careful analysis is required, thus, we leave this for future work.

Note that in our adaptive measurement technique, each bit of the difference vector $\delta_i$ indicates the difference of \emph{the whole syndromes} $\vec{s}_{i-1}$ and $\vec{s}_i$. One possible research direction would be seeing how an FTEC protocol can be improved by comparing each pair of bits of $\vec{s}_{i-1}$ and $\vec{s}_i$. Another possible direction would be improving the protocol by comparing syndromes from non-consecutive rounds (for example, $\vec{s}_{i-2}$ and $\vec{s}_i$).  
We point out that the standard FTEC method for a surface code of distance $d=2t+1$ performs $d$ rounds of syndrome measurements \cite{DKLP02}. If our adaptive measurement technique is modified using the aforementioned ideas, hopefully, we can obtain an FTEC protocol that the maximum number of rounds (not the maximum average number of rounds) grows linearly in $d$ similar to the FTEC protocol for a surface code, but is applicable to any stabilizer code. It might be harder to classically compute a syndrome for error correction from the measurement results, but the amount of required quantum resources may decrease, making the FTEC protocol more practical.

Last, we point out that the FTEC protocols developed in this work are applicable to any stabilizer code; we did not use any specific code structure in our development. It has been shown in \cite{DR20} that when tailored for specific families of codes, FTEC protocols with adaptive syndrome measurements can be further improved. We are hopeful that the FTEC protocols developed in this work could be further optimized for some families of codes as well.

\section{Acknowledgements}

We thank Shilin Huang and other members of Duke Quantum Center for helpful discussions on protocol development and possible future directions. The work was supported by the Office of the Director of National Intelligence - Intelligence Advanced Research Projects Activity through an Army Research Office contract (W911NF-16-1-0082), the Army Research Office (W911NF-21-1-0005), and the National Science Foundation Institute for Robust Quantum Simulation (QLCI grant OMA-2120757).

\appendix

\section{Proof of Theorem 2}
\label{sec:thm2_proof}

Let $\vec{\delta}$ be a difference vector obtained from repeated syndrome measurements, and suppose that there are at most $t$ faults in the FTEC protocol. First, consider the case that all bits in $\vec{\delta}$ are ones. The repeated syndrome measurements in \cref{pro:strong_any_t} will stop when $\mathrm{length}(\vec{\delta})=2t$, i.e., the number of rounds is $2t+1$. In this case, there are $I(2)$, $I(4)$, ..., $I(2t)$ faults, thus the syndrome $\vec{s}_{2t+1}$ is correct and can be used for error correction.

Next, consider the case that at least one bit in $\vec{\delta}$ is zero. We will show that the maximum length of $\vec{\delta}$ such that a usable zero substring does not exist (i.e., \cref{alg:find_usable} cannot return any usable zero substring when $t_\mathrm{in}=t$ and $\vec{\delta}_\mathrm{in}=\vec{\delta}$) is $\left(\frac{t+3}{2}\right)^2-3$ when $t$ is odd, and is $\left(\frac{t+2}{2}\right)\left(\frac{t+4}{2}\right)-3$ when $t$ is even. If this is true, we know that whenever $\mathrm{length}(\vec{\delta})=\left(\frac{t+3}{2}\right)^2-2$ and $t$ is odd (or $\mathrm{length}(\vec{\delta})=\left(\frac{t+2}{2}\right)\left(\frac{t+4}{2}\right)-2$ and $t$ is even), at least one usable zero substring must be found by \cref{alg:find_usable}. The number of rounds that gives such a difference vector is $\left(\frac{t+3}{2}\right)^2-1$ when $t$ is odd (or $\left(\frac{t+2}{2}\right)\left(\frac{t+4}{2}\right)-1$ when $t$ is even). Note that $\left(\frac{t+3}{2}\right)^2-1 \geq 2t+1$ for any $t$, and $\left(\frac{t+2}{2}\right)\left(\frac{t+4}{2}\right)-1\geq 2t+1$ when $t\geq 2$.

Let $\vec{\delta}=\eta_1 1 \eta_2 1 \dots 1 \eta_c$ for some positive integer $c$, where $\eta_i$ ($i=1,\dots,c$) is a substring of the form $00\dots0$ (the length of $\eta_i$ can be zero). Also, let $\alpha_j$, $\beta_j$, $\gamma_j$ be defined as in \cref{thm:condition_for_usable} for each $\eta_j$ with positive length, $p$ be the total number of $\eta_j$'s with positive length, and $q$ be the total number of ones in $\vec{\delta}$. If all $\eta_j$'s with positive length are unusable, then $\alpha_j+\beta_j+\gamma_j \leq t-1$ for any $j$ by \cref{thm:condition_for_usable}. 

For each $p$ and $q$, we will try to find an arrangement of zero substrings in $\vec{\delta}$ such that $\alpha_j+\beta_j$ is at the minimum for all $\eta_j$ with positive length (so that all $\gamma_j$'s and the length of $\vec{\delta}$ are at the maximum). This can be done by placing $\eta_j$'s as close to one another as possible (any other arrangement with the same $p$ and $q$ will give $\vec{\delta}$ with the same or smaller length). In this proof, we will consider the following three major arrangements.

Case 1: there are no $\eta_1$ and $\eta_c$. $\vec{\delta}$ will be in the form,
\begin{equation}
	\underbrace{1 \; \eta_2 \; 1 \; \eta_3 \; 1 \dots 1 \; \eta_{p+1} \; 1}_{(p+1) \text{\;ones}} \underbrace{1\;1 \dots 1 \; 1}_{(q-p-1) \text{\;ones}} \nonumber
\end{equation}

(1.a) If $q-p-1 \geq 0$ is even, the constraint for each $\gamma_j$ is,
\begin{equation}
	\alpha_j+\beta_j+\gamma_j = p-1+\frac{q-p-1}{2} +\gamma_j \leq t-1, \nonumber
\end{equation}
or equivalently,
\begin{equation}
	\gamma_j \leq t-\frac{p}{2}-\frac{q}{2}+\frac{1}{2}\;\text{when} \;j=1,\dots,p. \nonumber
\end{equation}
Let $f(p,q)$ be the maximum length of $\vec{\delta}$ for each $p$ and $q$. In this case,
\begin{equation}
	f(p,q)=\left(t-\frac{p}{2}-\frac{q}{2}+\frac{1}{2}\right)p+q, \label{eq:1}
\end{equation}	
which is valid when $p \geq 1$.

(1.b) If $q-p-1 \geq 1$ is odd, the constraint for each $\gamma_j$ is,
\begin{align}
	\gamma_j &\leq t-\frac{p}{2}-\frac{q}{2}+1\;\text{when} \;j=2,\dots,p, \nonumber\\
	\gamma_{p+1} &\leq t-\frac{p}{2}-\frac{q}{2}. \nonumber
\end{align}
In this case,
\begin{align}
	f(p,q)=&\left(t-\frac{p}{2}-\frac{q}{2}+1\right)(p-1) \nonumber \\
	&+\left(t-\frac{p}{2}-\frac{q}{2}\right)+q. \label{eq:2}
\end{align}
which is valid when $p \geq 1$.

Case 2: there is $\eta_1$ but no $\eta_c$. $\vec{\delta}$ will be in the form,
\begin{equation}
	\underbrace{\eta_1 \; 1 \; \eta_2 \; 1 \; \eta_3 \; 1 \dots 1 \; \eta_p \; 1}_{p \text{\;ones}} \underbrace{1\;1 \dots 1 \; 1}_{(q-p) \text{\;ones}} \nonumber
\end{equation}
(2.a) if $q-p \geq 0$ is even, the constraint for each $\gamma_j$ is,
\begin{align}
	\gamma_{1} &\leq t-\frac{p}{2}-\frac{q}{2}, \nonumber\\
	\gamma_j &\leq t-\frac{p}{2}-\frac{q}{2}+1\;\text{when} \;j=2,\dots,p. \nonumber
\end{align}
In this case,
\begin{align}
	f(p,q)=&\left(t-\frac{p}{2}-\frac{q}{2}+1\right)(p-1) \nonumber\\
	&+\left(t-\frac{p}{2}-\frac{q}{2}\right)+q. \nonumber
\end{align}
which is valid when $p \geq 1$. Note that this equation is the same as \cref{eq:2}.

(2.b) if $q-p \geq 1$ is odd, the constraint for each $\gamma_j$ is,
\begin{align}
	\gamma_{j} &\leq t-\frac{p}{2}-\frac{q}{2}+\frac{1}{2}\;\text{when}\;j=1,p, \nonumber\\
	\gamma_j &\leq t-\frac{p}{2}-\frac{q}{2}+\frac{3}{2}\;\text{when}\;j=2,\dots,p-1. \nonumber
\end{align}
If $p=1$, $f(1,q)=\left(t-\frac{q}{2}\right)+q$, which is similar to $f(p,q)$ in \cref{eq:1} with $p=1$. If $p\geq 2$, we have,
\begin{align}
	f(p,q)=&\left(t-\frac{p}{2}-\frac{q}{2}+\frac{3}{2}\right)(p-2) \nonumber \\
	&+2\left(t-\frac{p}{2}-\frac{q}{2}+\frac{1}{2}\right)+q. \label{eq:3}
\end{align}

Case 3: there are both $\eta_1$ and $\eta_c$. $\vec{\delta}$ will be in the form,
\begin{equation}
	\underbrace{\eta_1 \; 1 \; \eta_2 \; 1 \; \eta_3 \; 1 \dots 1 \; \eta_{p-1} \; 1}_{(p-1) \text{\;ones}} \underbrace{1\;1 \dots 1 \; 1}_{(q-p+1) \text{\;ones}}\; \eta_c \nonumber
\end{equation}
Note that this case is possible only when $p\geq 2$.

(3.a) If $q-p+1 \geq 0$ is even, the constraint for each $\gamma_j$ is,
\begin{align}
	\gamma_{j} &\leq t-\frac{p}{2}-\frac{q}{2}+\frac{1}{2}\;\text{when}\;j=1,c, \nonumber\\
	\gamma_j &\leq t-\frac{p}{2}-\frac{q}{2}+\frac{3}{2}\;\text{when}\;j=2,\dots,p-1. \nonumber
\end{align}
In this case,
\begin{align}
	f(p,q)=&\left(t-\frac{p}{2}-\frac{q}{2}+\frac{3}{2}\right)(p-2) \nonumber \\
	&+2\left(t-\frac{p}{2}-\frac{q}{2}+\frac{1}{2}\right)+q, \nonumber
\end{align}
which is valid when $p \geq 2$. Note that this equation is the same as \cref{eq:3}.

(3.b) If $q-p+1 \geq 1$ is odd, the constraint for each $\gamma_j$ is,
\begin{align}
	\gamma_{j} &\leq t-\frac{p}{2}-\frac{q}{2}+1\;\text{when}\;j=1,p-1, \nonumber\\
	\gamma_j &\leq t-\frac{p}{2}-\frac{q}{2}+2\;\text{when}\;j=2,\dots,p-2, \nonumber\\
	\gamma_{c} &\leq t-\frac{p}{2}-\frac{q}{2}. \nonumber
\end{align}
If $p=2$, $f(2,q)=2\left(t-\frac{1}{2}-\frac{q}{2}\right)+1+q$, which is similar to $f(p,q)$ in \cref{eq:2} with $p=2$. If $p\geq 3$, we have,
\begin{align}
	f(p,q)=&\left(t-\frac{p}{2}-\frac{q}{2}+2\right)(p-3) \nonumber\\
	&+2\left(t-\frac{p}{2}-\frac{q}{2}+1\right) \nonumber \\
	&+\left(t-\frac{p}{2}-\frac{q}{2}\right)+q. \label{eq:4}
\end{align}

Next, we will find the maximum value of $f(p,q)$ from each of \cref{eq:1,eq:2,eq:3,eq:4}. In any case, we find that $\frac{\partial f(p,q)}{\partial q}=-\frac{p}{2}+1$. Therefore,
\begin{enumerate}
	\item if $p=1$, the maximum value of $f(1,q)$ is attained at the largest possible value of $q$;
	\item if $p=2$, $f(2,q)$ is the same for all possible values of $q$;
	\item if $p\geq 3$, the maximum value of $f(p,q)$ is attained at the smallest possible value of $q$.
\end{enumerate}
Note that for \cref{eq:1,eq:3}, $p$ and $q$ must have different parities. In contrast, for \cref{eq:2,eq:4}, $p$ and $q$ must have the same parity. 

When $p=1$ or $2$, the maximum value of $f(p,q)$ is $2t-1$. When $p\geq 3$, The maximum value of $f(p,q)$ in each case is as follows:
\begin{enumerate}
	\item From \cref{eq:1}, $f(p_\mathrm{max},q_\mathrm{max})=\left(\frac{t+1}{2}\right)^2+1$ (or $\frac{t}{2}\left(\frac{t}{2}+1\right)+1$) at $q_\mathrm{max}=p_\mathrm{max}+1$ and $p_\mathrm{max}=\frac{t+1}{2}$ when $t$ is odd (or $p_\mathrm{max}=\frac{t}{2}$ when $t$ is even).
	\item From \cref{eq:2}, $f(p_\mathrm{max},q_\mathrm{max})=\left(\frac{t+1}{2}\right)\left(\frac{t+3}{2}\right)-1$ (or $\left(\frac{t}{2}+1\right)^2-1$) at $q_\mathrm{max}=p_\mathrm{max}$ and $p_\mathrm{max}=\frac{t+1}{2}$ when $t$ is odd (or $p_\mathrm{max}=\frac{t}{2}+1$ when $t$ is even).
	\item From \cref{eq:3}, $f(p_\mathrm{max},q_\mathrm{max})=\left(\frac{t+3}{2}\right)^2-3$ (or $\left(\frac{t}{2}+1\right)\left(\frac{t}{2}+2\right)-3$) at $q_\mathrm{max}=p_\mathrm{max}-1$ (which is possible in Case 3.a) and $p_\mathrm{max}=\frac{t+3}{2}$ when $t$ is odd (or $p_\mathrm{max}=\frac{t}{2}+1$ when $t$ is even).
	\item From \cref{eq:4}, $f(p_\mathrm{max},q_\mathrm{max})=\left(\frac{t+3}{2}\right)^2-4$ (or $\left(\frac{t}{2}+1\right)\left(\frac{t}{2}+2\right)-4$) at $q_\mathrm{max}=p_\mathrm{max}$ and $p_\mathrm{max}=\frac{t+3}{2}$ when $t$ is odd (or $p_\mathrm{max}=\frac{t}{2}+1$ when $t$ is even).
\end{enumerate}
Combining all possible cases, the maximum value of $f(p,q)$ which is the maximum length of $\vec{\delta}$ with no usable zero substrings is $\left(\frac{t+3}{2}\right)^2-3$ (or $\left(\frac{t+2}{2}\right)\left(\frac{t+4}{2}\right)-3$). When $t$ is odd, an explicit form of $\vec{\delta}$ of the maximum length is $\eta_1 1 \eta_2 1\ \dots 1 \eta_{\frac{t+3}{2}}$ with $\gamma_1,\gamma_{\frac{t+3}{2}} = \frac{t-1}{2}$ and $\gamma_j = \frac{t+1}{2}$ ($j=2,\dots,\frac{t-1}{2}$). When $t$ is even, an explicit form of $\vec{\delta}$ of the maximum length is $\eta_1 1 \eta_2 1\ \dots 1 \eta_{\frac{t}{2}+1}$ with $\gamma_1,\gamma_{\frac{t}{2}+1} = \frac{t}{2}$ and $\gamma_j = \frac{t}{2}+1$ ($j=2,\dots,\frac{t}{2}$). 

For this reason, performing syndrome measurements $\left(\frac{t+3}{2}\right)^2-1$ rounds when $t$ is odd (or $\left(\frac{t+2}{2}\right)\left(\frac{t+4}{2}\right)-1$ rounds when $t$ is even) can guarantee that a usable zero substring exists. The error correction can be done using the syndrome corresponding to the usable zero substring.

\section{Classical time complexity of Algorithm 1}
\label{sec:complexity}

In this section, we analyze the classical time complexity of \cref{alg:find_usable} which is called after each round of syndrome measurements in both adaptive strong and adaptive weak decoders. From \cref{thm:bound_strong}, the length of the difference vector $\vec{\delta}$ calculated after any round of syndrome measurements is at most $\left(\frac{t+3}{2}\right)^2-2$ when $t$ is odd (or at most $\left(\frac{t+2}{2}\right)\left(\frac{t+4}{2}\right)-2$ when $t$ is even). For each $\eta_j$ in $\vec{\delta}$, $\alpha_j$, $\beta_j$, and $\gamma_j$ are calculated by counting the number of non-overlapping 11 substrings plus the number of remaining one bits in $\vec{\delta}$, and the number of zeros in $\eta_j$. These calculations can be done by sweeping through the entire length of $\vec{\delta}$, so the number of required operations is linear in the length of $\vec{\delta}$. After the calculations, the condition $\alpha_j+\beta_j+\gamma_j \geq t$ which requires a constant number of operations is checked. Thus, computing whether each $\eta_j$ is usable requires no more than $K_1\left[\left(\frac{t+3}{2}\right)^2-2\right]+K_2$ operations where $K_1,K_2$ are some constant numbers.

For any $t$, the total number of one bits in $\vec{\delta}$ is at most $2t$ since each fault causes at most two ones in $\vec{\delta}$. Because $\vec{\delta}$ is of the form $\eta_1 1 \eta_2 1 \dots 1 \eta_c$, the maximum number of $\eta_j$ is $2t+1$. Therefore, total number of operations in \cref{alg:find_usable} is no more than $K_1\left[\left(\frac{t+3}{2}\right)^2-2\right](2t+1)+K_2(2t+1) = O(t^3)$.

\bibliographystyle{quantum}
\bibliography{bibtex_TPB23}

\end{document}